\documentclass[11pt,report]{ETreport}
\usepackage{float,amsmath,amsfonts,amssymb,amsthm,graphicx}

\usepackage{ marvosym }	

\usepackage{array}
\usepackage{amssymb,amsmath,amsfonts,amsthm}
\usepackage{color}
\usepackage{tikz,pgfplots}

\newfloat{Algorithm}{htb}{loa}

\newcommand{\dd}{\mathrm{d}}
\newcommand{\bbR}{\mathbb{R}}
\newcommand{\normal}{\mathrm{N}}
\newcommand{\tsfrac}[2]{{\textstyle \frac{#1}{#2}}}
\newcommand{\twobyone}[2]{\begin{array}{c} #1 \\ #2 \end{array}}
\newcommand{\twobytwo}[4]{\begin{array}{cc} #1 & #2 \\ #3 & #4 \end{array}}
\newcommand{\nrm}[1]{\Vert #1 \Vert}

\newtheoremstyle{egal}
 {.5em}
 {.5em}
 {}
 {}
 {\bf}
 {}
 {1.5ex}
 {}

\theoremstyle{egal}
\newtheorem{theorem}{Theorem}[section]

\newtheorem{corollary}{Corollary}[section]

\newtheorem{lemma}{Lemma}[section]

\newtheorem{assumption}[theorem]{Assumption}

\title{Metropolis-Hastings algorithms with autoregressive proposals,\\ and a few examples}
\author{Richard A. Norton \& Colin Fox}      
\date{}     

\ETRnum{2016-1}     
\ETRprintDate{May 2016}            
\ETRauthor{Richard A. Norton, Colin Fox}       

\ETRauthorEmail{richard.norton@otago.ac.nz}

\begin{document}
\maketitle

\begin{abstract}
We analyse computational efficiency of Metropolis-Hastings algorithms with stochastic AR(1) process proposals.  These proposals include, as a subclass, discretized Langevin diffusion (e.g. MALA) and discretized Hamiltonian dynamics (e.g. HMC).  

We derive expressions for the expected acceptance rate and expected jump size for MCMC methods with general stochastic AR(1) process proposals for the case where the target distribution is absolutely continuous with respect to a Gaussian and the covariance of the Gaussian is allowed to have off-diagonal terms.  This allows us to extend what is known about several MCMC methods as well as determining the efficiency of new MCMC methods of this type.  In the special case of Hybrid Monte Carlo, we can determine the optimal integration time and the effect of the choice of mass matrix.

By including the effect of Metropolis-Hastings we also extend results by Fox and Parker, who used matrix splitting techniques to analyse the performance and improve efficiency of stochastic AR(1) processes for sampling from Gaussian distributions.

\end{abstract}

\cleardoublepage        
\tableofcontents

\ETRmain
\chapter{Introduction}
\label{sec intro}

We consider Metropolis-Hastings (MH) algorithms for sampling from a target distribution $\pi_d$ using a stochastic AR(1) process proposal; given current state $x \in \bbR^d$ the proposal $y \in \bbR^d$ is given by
\begin{equation}
\label{eq ar1}
	y = G x + g + \nu
\end{equation}
where $G \in \bbR^{d\times d}$ is the iteration matrix, $g \in \bbR^d$ is a fixed vector and $\nu$ is an independent and identically distributed (i.i.d.) draw from $\normal(0,\Sigma)$.  In general, $G$, $g$ and $\Sigma$ may depend on $x$.  We will refer to \eqref{eq ar1} as an \emph{AR(1) proposal} or \emph{stochastic AR(1) proposal}.  The proposal is accepted with probability 
$$
	\alpha(x,y) = 1 \wedge \frac{\pi_d(y) q(y,x)}{\pi_d(x) q(x,y)}
$$
where $\pi_d(x)$ denotes the target probability density function, $q(x,\dd y) = q(x,y) \dd y$ is the transition kernel for the proposal $y$ given current state $x$, and $p \wedge q = \min\{p,q\}$.  

Algorithms using AR(1) proposals include: the random-walk Metropolis algorithm (RWM) \cite{RGG1997}, the simplified Langevin algorithm (SLA) \cite{BRS2009}, the so-called $\theta$-SLA method \cite{BRS2009}, preconditioned versions of RWM and SLA \cite{BRS2009}, and the Crank-Nicolson (CN) and preconditioned Crank-Nicolson (pCN) proposals \cite{CRSW2013}.  When the target distribution is Gaussian, then the Metropolis-adjusted Langevin algorithm (MALA) \cite{RR1998} and the Hybrid Monte Carlo algorithm (HMC) \cite{DKPR1987,BPRSS2010,N1993} can also be written in the form of \eqref{eq ar1}, and MALA is the same as SLA.  For any target distribution, one step of HMC is the same as MALA \cite{BPRSS2010}.

Analysis of these Markov chain Monte Carlo (MCMC) methods is almost exclusively limited to the case when the target distribution $\pi_d$ is a change of measure from a reference product distribution so that 
\begin{equation}
\label{eq target0}
	\frac{\dd \pi_d}{\dd \tilde{\pi}_d}(x) = \exp( - \phi_d(x))
\end{equation}
for some $\phi_d: \bbR^d \mapsto \bbR$ where $\tilde{\pi}_d$ is a product distribution of the form
\begin{equation}
\label{eq target1}
	\tilde{\pi}_d(x) = \prod_{i=1}^d \lambda_i f\left( \lambda_i x_i  \right)
\end{equation}
for some $f : \bbR \mapsto \bbR$ and sequence $\{ \lambda_i \}_{i=1}^d \subset \bbR$.  We can view $\pi_d$ and $\tilde{\pi}_d$ as finite-dimensional approximations of infinite-dimensional measures $\pi$ and $\tilde{\pi}$ satisfying $\tsfrac{\dd \pi}{\dd \tilde{\pi}}(x) = \exp( -\phi(x))$ on some state space.  

Examples of inverse problems that yield posterior distributions of this form can be found in \cite{S2010,BRS2009,BRSV2008}.


Analysis of MH algorithms with AR(1) proposals is also typically limited to the cases when $d=1$ or $d \rightarrow \infty$.  Of which the $d \rightarrow \infty$ case is more important because it is used as an approximation for the practical computational problem when $d$ is large but finite.  Precisely what is meant by `large' is problem dependent.  For example, \cite[\S 3]{RR1998} demonstrates that for MALA with $\phi_d = 0$, $f(x) \propto \exp(-x^2/2)$ and $\lambda_i=1$, then $d=5$ is large, but if $f$ is non-symmetric then large $d$ is greater than $10$.  

Analyses of RWM and MALA began with the case when $\phi = 0$ and $\lambda_i = 1$ for all $i$ (so that $\pi_d$ has product form with i.i.d. components) based on discretizations of a Langevin diffusion process.  Roberts, Gelman and Gilks \cite{RGG1997} for RWM and then Roberts and Rosenthal \cite{RR1998} for MALA showed that as $d \rightarrow \infty$, the first component of the Markov chain converges to a Langevin diffusion process and the `speed' of the diffusion process is maximised when the acceptance rate is $0.234$ for RWM and $0.574$ for MALA.  This is equivalent to maximising the expected squared jump size of the Markov chain.  The non i.i.d. cases when $\phi = 0$ and $\lambda_i \neq 1$ for RWM and MALA are subsequently treated in \cite{RR2001,B2007,B2008,BR2008}.  In these articles it is noted that while the optimal acceptance rate for RWM and MALA remains the same, the expected jump size of RWM decreases as the $l^2$-norm of the sequence $\{ \lambda_i \}$ increases, while MALA depends on the $l^6$-norm.  An `inhomogeneous' RWM proposal is also considered in \cite{B2007}, which is what we will call `preconditioning'.  In \cite{BR2008}, it is also noted that their results also hold when the target distribution is a multivariate normal whatever the covariance matrix since orthogonal transformations can transform the target to one with independent components.  We will exploit this fact throughout this article.  A non-product form of target for RWM is considered in \cite{BR2000}.  

The case of RWM and SLA (a simplified version of MALA) for non-product target distributions, when $\phi \neq 0$, is considered in \cite{BRS2009}.  The optimal acceptance rates remain $0.234$ for RWM and $0.574$ for SLA (same as MALA), under certain conditions on $\phi_d$ and $\{ \lambda_i \}_{i=1}^\infty$.  Again, the expected jump sizes for these algorithms decrease with the $l^2$- and $l^6$-norms of $\{ \lambda_i \}$.  It is also suggested how to precondition the RWM and SLA proposals in \cite{BRS2009}.  

HMC was analysed in \cite{BPRSS2010} for the case when the target distribution has product form with i.i.d. components.  Similarly to the analyses of RWM, MALA and SLA, the authors of \cite{BPRSS2010} showed that the expected squared jump size for HMC is maximised when the acceptance rate is $0.651$, and this corresponds to $\mathcal{O}(d^{1/4})$ steps to traverse state space.  This compares favourably with RWM and MALA which require $\mathcal{O}(d)$ and $\mathcal{O}(d^{1/3})$ steps respectively for the same problem, but still blows up as $d \rightarrow \infty$.

More recent analyses has shown that some methods can be modified so they are well-defined in the infinite-dimensional function space setting, so in the limit as $d \rightarrow \infty$, the methods achieve a positive acceptance rate without zero step size and only $\mathcal{O}(1)$ steps are required to traverse state space.  This requires modifying the proposals by `preconditioning' and/or a coordinate transformation.  For example, in the case when the target is a change of measure from a Gaussian reference measure, when
\begin{equation}
\label{eq ref}
	\tilde{\pi}_d(x) \propto \exp \left( - \frac{1}{2} x^T A x + b^T x \right),
\end{equation}
for symmetric positive definite matrix $A \in \bbR^{d \times d}$ and vector $b \in \bbR^d$, the CN and pCN proposals are analysed in \cite{CRSW2013}, and a variant of HMC is analysed in \cite{BPSSS2011}.  Other examples of this approach include \cite{BRSV2008,BS2009a,PST2012}.  Another way to view preconditioning of MALA and HMC is given by \cite{GC2011}.

In pCN and the variant of HMC in \cite{BPSSS2011} it is necessary to draw independent samples from $\normal(0,A^{-1})$ or compute a spectral decomposition of $A$ which effectively transforms the reference measure to a product distribution; both of which could be computationally infeasible when $d$ is large.  In CN, the action of $(I + t A)^{-1}$ for some $t >0$ is required per iteration of the Markov chain, see \cite{CRSW2013}, which may be expensive to compute.

The results in \cite{BRS2009}, where $\tilde{\pi}_d$ has product form, easily extend to the case where $\tilde{\pi}_d$ is Gaussian \eqref{eq ref} where $A$ may have non-zero off-diagonal terms.  This is obvious once we recognise that the Markov chains for RWM and SLA are invariant to orthogonal coordinate transformations.  That is, there exists an orthogonal coordinate transformation that diagonalizes the covariance matrix of $\pi_d$ and $G$ and $\Sigma$ in \eqref{eq ar1}, see e.g. Lemma \ref{lem2}.  It is important to note that it is not necessary to compute the orthogonal transformation, as it is enough to simply know that it exists, and the efficiency of the untransformed Markov chain is identical to the transformed chain.

We will extend this idea to MH algorithms with general AR(1) proposals where $G$ and $\Sigma$ are functions of $A$, targeting distributions that are either Gaussian, or a change of measure from a reference Gaussian distribution \eqref{eq ref}.  In particular, the Gaussian reference distribution is allowed to have off-diagonal terms so it is not restricted to product form.  

Therefore, we extend the study of MH algorithms with particular AR(1) proposals to general AR(1) proposals where $G$ and $\Sigma$ are functions of the reference precision matrix $A$.  We also extend the study of MH algorithms with AR(1) proposals targeting distributions defined by \eqref{eq target0} and \eqref{eq target1} to target distributions defined by \eqref{eq target0} and \eqref{eq ref} where $A$ may have off-diagonal terms.


Another important feature of this analysis is that the proposals do not necessarily require independent samples from $\normal(0,A^{-1})$ or $\normal(0,A)$, multiplying by $A^{1/2}$ or $A^{-1/2}$, or computing a spectral decomposition of $A$ or $A^{-1}$; even though we use the existence of a spectral decomposition of $A$ for theoretical purposes.  This fact separates this new theory from previous theory for CN, pCN and HMC in \cite{CRSW2013,BPSSS2011}, where these proposals include operations that may be computationally infeasible in high dimensions.  

By generalising the results in \cite{BRS2009,BPRSS2010} under some assumptions, we calculate limits for the expected acceptance rate and expected squared jump size for MH algorithms with AR(1) proposals as $d \rightarrow \infty$.  We can then decide on the efficiency of a method based on expected jump size and the computing cost for each proposal.

Our new theory encompasses existing MCMC methods with AR(1) proposals, which are now special cases for our theory, and we can extend the results that are currently available for SLA, HMC, $\theta$-SLA, and preconditioned versions of these MCMC methods, see Section \ref{sec examples}.  
In the case of HMC for a Gaussian target, we are no longer restricted, as in \cite{BPRSS2010}, to an i.i.d. product target, and we now have criteria for how to choose the mass matrix (preconditioner) and the total time to integrate the Hamiltonian system.  Previous analyses of HMC only provided guidance on tuning the time step until the acceptance rate is $0.651$.

We can also apply our new theory to new MCMC methods.  For example we can analyse an MCMC method where the proposal is $L$ steps of the SLA proposal before accepting or rejecting.  We show that for any $L$, the step size should be tuned until the acceptance rate is $0.574$, the same as MALA and SLA, and when the computing cost is dominated by matrix-vector products with $A$, then it is optimal to use $L>1$.   Moreover, as the cost of evaluating $\phi_d$ increases, so does the optimal choice of $L$.  

Our analysis relies on the theory of \emph{matrix splitting} which originated in numerical linear algebra for iteratively solving linear systems of equations \cite{A1994}, but has since been applied to sampling from Gaussian distributions \cite{F2013,FP2016,FP2014}.  As we will see in Section \ref{sec equiv}, if the spectral radius of $G$ is less than $1$, then it is possible to rewrite \eqref{eq ar1} in terms of a matrix splitting of a matrix $\mathcal{A}$, which is not equal to $A$ in general.  By defining \emph{splitting} matrices $M$ and $N$ such that $\mathcal{A} = M-N$ then $y$ from \eqref{eq ar1} satisfies
\begin{equation}
\label{eq ar1b}
	M y = N x + \beta + \nu
\end{equation}
where $\beta$ is a vector, $\nu$ is an i.i.d. draw from $\normal(0,M^T+N)$, $G = M^{-1}N$, $g = M^{-1}\beta$ and $\Sigma = M^{-1}(M^T + N) M^{-T}$.  The converse statement, $y$ satisfies \eqref{eq ar1} if $y$ satisfies \eqref{eq ar1b}, only requires that $M^{-1}$ exists.  Moreover, if the spectral radius of $G$ is less than $1$, then the Markov chain generated by \eqref{eq ar1} or \eqref{eq ar1b} without the MH accept/reject step, will converge to $\normal(\mathcal{A}^{-1}\beta,\mathcal{A}^{-1})$, which we call the \emph{proposal limit distribution}, see \cite{FP2016}.  We call this Markov chain the \emph{proposal chain}.

Fox and Parker \cite{FP2016} realised that Gibbs sampling from a Gaussian $\normal(\mathcal{A}^{-1}\beta,\mathcal{A}^{-1})$ is very closely related to the Gauss-Seidel iterative solution to a linear system of equations $\mathcal{A}x = \beta$, and that Gauss-Seidel and Gibbs sampling use the same matrix splitting $\mathcal{A} = M-N$; $M = L + D$ and $N = U$ where $L$, $D$ and $U$ are the strictly lower triangular, diagonal, and upper triangular parts of $\mathcal{A}$ respectively.  A generalisation of this observation is that all proposal chains generated by \eqref{eq ar1b} are generalised fixed-scan Gibbs samplers for a Gaussian.

If \eqref{eq ar1} or \eqref{eq ar1b} is a proposal for the MH algorithm then the transition kernel changes from that of the proposal chain and we cannot use the theory in \cite{FP2016} to determine its convergence properties.  Moreover, acceleration techniques suggested in \cite{FP2014,FP2016} for the proposal chain may not accelerate the MH algorithm.  For example, Goodman and Sokal \cite{GS1989} accelerated Gibbs sampling of normal distributions using ideas from multigrid linear solvers, but only observed modest efficiency gains in their non-normal examples (exponential distributions with fourth moments).  Also, Green and Han \cite{GH1992} applied successive-over-relaxation to a local Gaussian approximation of a non-Gaussian target distribution as a proposal for the MH algorithm.  Again, they did not observe significant acceleration in the non-normal target case.

This article is useful for designing efficient MCMC methods because it shows how the expected squared jump size depends on the eigenvalues of $G$ and the difference between the proposal limit $\normal(\mathcal{A}^{-1}\beta,\mathcal{A}^{-1})$ and the target reference $\normal(A^{-1}b,A^{-1})$.  Generally, an efficient MH algorithm with an AR(1) proposal should satisfy:
\begin{enumerate}
\item $\mathcal{A} = A$ and $\beta = b$ or have small differences.
\item The spectral radius of $G$ should be bounded well below $1$.
\item The proposal should be cheap to compute, i.e. multiplying by $G$ or $M^{-1}$ and independent sampling from $\normal(0,\Sigma)$ or $\normal(0,M^T+N)$ should be cheap to compute.
\end{enumerate}
In addition to providing these `rules of thumb', we have quantified the effect of the eigenvalues of $G$, the difference between the eigenvalues of $A$ and $\mathcal{A}$, and the difference between $A^{-1}b$ and $\mathcal{A}^{-1}\beta$.

Our long-term goal is to construct an AR(1) proposal for the MH algorithm based on local Gaussian approximations to the target $\pi_d$, using ideas in this article.  For example, we might choose $A$ and $b$ such that $-\tsfrac{1}{2} x^T \mathcal{A} x + \beta^T x$ is a local quadratic approximation to $\log \pi_d$, and then choose $M$ and $N$ to define a proposal such that $\mathcal{A} = A$ and $\beta = b$ and the proposal is cheap to compute.  To a certain extent, choosing $A$ and $b$ to obtain a local quadratic approximation to $\log \pi_d$ mimics the design of some optimization algorithms, see e.g. \cite{NW1999}, so optimization theory could be a source of inspiration for designing sampling algorithms.


Our analysis is limited to cases where $M$ and $N$ are functions of $A$.  This allows us to simultaneously diagonalise both the AR(1) proposal and the target reference distribution with a coordinate transformation and define a parallel Markov chain (that we never compute with) such that it has the same convergence properties as the original.  This is not an overly restrictive condition if the dimension is high enough to make a spectral decomposition of $A$ impractical to compute.  We will see below that it includes several important examples of MH algorithms already in use.


The remaining sections are as follows.  In Section \ref{sec equiv} we show that \eqref{eq ar1} and \eqref{eq ar1b} are equivalent, then Section \ref{sec gaussian} presents new analyses for the expected acceptance rate and jump size of MH algorithms with AR(1) proposals when the target distribution is Gaussian ($\phi_d = 0$).  We then extend these results to the non-Gaussian case in Section \ref{sec nongaussian}.  Section \ref{sec examples} then applies this new analysis to proposals from Langevin diffusion and Hamiltonian dynamics.  We see that these proposals are AR(1) proposals and we identify the corresponding matrix splitting and proposal limit distribution.  Using our earlier analysis we assess the convergence properties of these methods as $d \rightarrow \infty$.  We provide concluding remarks in Section \ref{sec conclusion}.  Several proofs have been moved to the Appendix to improve readability.

\chapter{Preliminary results and notation}

\section{Stochastic AR(1) processes correspond to matrix splittings}
\label{sec equiv}

We can express a stochastic AR(1) process using either \eqref{eq ar1} or \eqref{eq ar1b}, provided it converges.  

\begin{theorem}
\label{thm:2.1}
If we are given $G$, $g$ and $\Sigma$, and the spectral radius of $G$ is less than $1$, then the stochastic AR(1) process \eqref{eq ar1} can be written as \eqref{eq ar1b} using
\begin{equation}
\label{eq calA}
	\mathcal{A} = \left( \sum_{l=0}^\infty G^l \Sigma (G^{T})^l \right)^{-1}
\end{equation}
and
\begin{equation}
\label{eq MNbeta}
\begin{split}
	M &= \mathcal{A} (I-G)^{-1}, \\
	N &= \mathcal{A} (I-G)^{-1} G, \qquad \qquad 
	 \beta = \mathcal{A} (I-G)^{-1} g.
\end{split}
\end{equation}
Note that $\mathcal{A} = M - N$ is symmetric and positive definite (spd).
\end{theorem}

\begin{proof}
Since the spectral radius of $G$ is less than $1$ and $\Sigma$ is spd, it follows that $\mathcal{A}^{-1} := \sum_{l=0}^\infty G^l \Sigma (G^{T})^l$ is well-defined and spd.  Then \eqref{eq calA} and \eqref{eq MNbeta} satisfy $\mathcal{A} = M-N$, $G = M^{-1} N$ and $g = M^{-1} \beta$.  We must also check that $\Sigma = M^{-1} (M^T + N) M^{-T}$ is satisfied.  Substituting \eqref{eq MNbeta} into $M^{-1} (M^T+N)M^{-T}$ we get
\begin{align*}
	M^{-1}(M^T + N) M^{-T} 
	&= (I-G) \mathcal{A}^{-1} + G \mathcal{A}^{-1} G^T \\
	&= \mathcal{A}^{-1} - G \mathcal{A}^{-1} G^T \\
	&= \sum_{l=0}^\infty G^l \Sigma (G^T)^l - \sum_{l=1}^\infty G^l \Sigma (G^T)^l \\
	&= \Sigma.
\end{align*}
\end{proof}

If we are given $M$, $N$ and $\beta$, and $M^{-1}$ exists, then it is obvious that if $y$ satisfies \eqref{eq ar1b}, then $y$ satisfies \eqref{eq ar1} with $G = M^{-1}N$, $g = M^{-1}\beta$ and $\Sigma = M^{-1}(M^T + N) M^{-T}$. 

We also remark that Theorem \ref{thm:2.1} does not apply to RWM since $G = I$ for RWM.

In the following special case we obtain a symmetric matrix splitting.  

\begin{corollary}
\label{lem equiv}
If the spectral radius of $G$ is less than $1$ and $G \Sigma$ is symmetric, then the stochastic AR(1) process \eqref{eq ar1} has a corresponding matrix splitting defined by
\begin{align*}
	M &= \Sigma^{-1} (I + G), & \mathcal{A} &= M (I-G) = \Sigma^{-1} (I-G^2), \\
	N &= MG = \Sigma^{-1}(I+G)G, & \beta &= Mg = \Sigma^{-1}(I+G) g,
\end{align*}
and $M$ and $N$ are symmetric (we say the matrix splitting is symmetric).
\end{corollary}

\begin{proof}
These matrix splitting formulae follow from the identity 
$$
	\sum_{l=0}^\infty G^l \Sigma (G^T)^l = \sum_{l=0}^\infty G^{2l} \Sigma = (I-G^2)^{-1} \Sigma.
$$
To see that $M$ is symmetric (and hence also $N$ since $\mathcal{A}$ is symmetric and $\mathcal{A} = M-N$) we note that
$$
	M = \left( (I-G) \sum_{l=0}^\infty G^{2l} \Sigma \right)^{-1}.
$$
\end{proof}

\section{Notation}
\label{sec notation}

Throughout the remainder of this article we will use the following notation.  Let $G_i$, $\lambda_i^2$ and $\tilde{\lambda}_i^2$ be eigenvalues of $G$, $A$ and $\mathcal{A}$ respectively.  Also define
\begin{align*}
	\tilde{g}_i &:= 1-G_i, & g_i &:= 1-G_i^2, & m_i &:= (A^{-1}b)_i, & \tilde{m}_i &:= (\mathcal{A}^{-1}\beta)_i, \\
	r_i &:= \frac{\lambda_i^2-\tilde{\lambda}_i^2}{\lambda_i^2}, &
	\tilde{r}_i &:= \frac{\lambda_i^2}{\tilde{\lambda}_i^2}, & 
	\hat{r}_i &:= m_i - \tilde{m}_i
\end{align*}
and 
\begin{align*}
	T_{0i} &:= \hat{r}_i^2 \lambda_i^2 ( \tsfrac{1}{2} r_i g_i - \tilde{g}_i), &
	T_{1i} &:= \hat{r}_i \lambda_i ( r_i g_i - \tilde{g}_i), &
	T_{2i} &:= \hat{r}_i \lambda_i (\tilde{r}_i g_i )^{1/2} ( 1 - r_i G_i), \\
	T_{3i} &:= \tsfrac{1}{2} r_i g_i, &
	T_{4i} &:= -\tsfrac{1}{2} r_i \tilde{r}_i g_i, &
	T_{5i} &:= -r_i G_i (\tilde{r}_i g_i)^{1/2}.
\end{align*}
In general, all of these quantities may depend on $d$. The standard normal cumulative distribution function will always be $\Phi$.

We will say that $f_{d,i} = \mathcal{O}(g_{d,i})$ (uniformly in $i$) as $d \rightarrow \infty$ if for all $i$ and all sufficiently large $d$, $f_{d,i}/g_{d,i}$ is bounded by a constant that is independent of $d$ and $i$. Likewise, $f_{d,i} = \mathrm{o}(g_{d,i})$ (uniformly in $i$) as $d \rightarrow \infty$ if $\max_{1\leq i \leq d} f_{d,i}/g_{d,i} \rightarrow 0$ as $d \rightarrow \infty$.  For brevity we will sometimes omit ``uniformly in $i$''.

Other articles use $\lambda_i^2$ to denote the eigenvalues of the covariance matrix corresponding to $\tilde{\pi}_d$ \cite{BRS2009,BS2009b,BS2009a,BPSSS2011}.  We do not follow this convention and instead use $\lambda_i^2$ to denote eigenvalues of the precision matrix.  Since sampling from the Gaussian $N(A^{-1}b,A^{-1})$ is equivalent to solving the linear system $Ax=b$ (see \cite{F2013,FP2014,FP2016}), our notation aligns with literature on solving linear systems.


\chapter{Gaussian targets}
\label{sec gaussian}

\section{Expected acceptance rate for a Gaussian target}
\label{sec expect}

The expected acceptance rate is a quantity that is related to efficiency and for optimal performance the proposal is usually tuned so that the observed average acceptance rate is between $0$ and $1$.  For example, it has been shown, for particular target distributions, in the case when $d \rightarrow \infty$, that $0.234$ is optimal for RWM \cite{RGG1997}, $0.574$ is optimal for MALA \cite{RR1998} and SLA \cite{BRS2009}, and $0.651$ is optimal for HMC \cite{BPRSS2010}.  All of these results required expressions for the expected acceptance rate of the algorithm as $d \rightarrow \infty$.  Here we derive an expression for the expected acceptance rate for a MH algorithm with an AR(1) proposal \eqref{eq ar1b} and Gaussian target $\normal(A^{-1}b,A^{-1})$, provided the splitting matrices are functions of $A$.  Thus, our MH algorithm is defined by
\begin{equation}
\label{MH0}
\begin{split}
	\mbox{Target:} & \qquad \normal(A^{-1}b, A^{-1}), \\
	\mbox{Proposal:} & \qquad y = G x + M^{-1} \beta + (\mathcal{A}^{-1} - G \mathcal{A}^{-1} G^T)^{1/2} \xi, 
\end{split}
\end{equation}
where $\xi \sim \normal(0,I)$, and we have used $G = M^{-1}N$ and $M^{-1}(M^T+N)M^{-T} = \mathcal{A}^{-1} - G\mathcal{A}^{-1} G^T$ \cite[Lem. 2.3]{FP2014}.  The following lemma is a result of simple algebra.  The proof is in the Appendix.

\begin{lemma}
\label{lem1}
Suppose $\mathcal{A} = M-N$ is a symmetric splitting.  Then the acceptance probability for \eqref{MH0} satisfies
$$
	\alpha(x,y) = 1 \wedge \exp \left( -\frac{1}{2} y^T (A - \mathcal{A}) y + \frac{1}{2} x^T(A-\mathcal{A}) x + (b-\beta)^T (y-x)\right).
$$
\end{lemma}

Since $A$ is real and symmetric, we can define a spectral decomposition
$$
	A = Q \Lambda Q^T
$$ 
where $Q \in \bbR^{d\times d}$ is orthogonal and $\Lambda = \operatorname{diag}(\lambda_1^2,\dotsc,\lambda_d^2)$ is a diagonal matrix of eigenvalues of $A$.  Although we may not be able to compute $Q$ and $\Lambda$ this does not stop us from using the existence of a spectral decomposition for theory.  Simple algebra gives us the following result.

\begin{lemma}
\label{lem2}
Suppose $M = M(A)$ and $N = N(A)$ are functions of $A$.  Then $G$ and $\mathcal{A}$ are also functions of $A$ and under the coordinate transformation
$$
	x \leftrightarrow Q^T x
$$
the MH algorithm \eqref{MH0} is transformed to a MH algorithm defined by
\begin{equation}
\label{MH1}
\begin{split}
	\mbox{Target:} & \qquad \normal(\Lambda^{-1} Q^T b, \Lambda^{-1}), \\
	\mbox{Proposal:} & \qquad y = G x + M(\Lambda)^{-1} Q^T \beta + (\tilde{\Lambda}^{-1} - G \tilde{\Lambda}^{-1} G^T)^{1/2} \xi, 
\end{split}
\end{equation}
where $\xi \sim \normal(0,I)$, and $G = M(\Lambda)^{-1}N(\Lambda)$ and $\tilde{\Lambda} = \mathcal{A}(\Lambda)$ are diagonal matrices.  
\end{lemma}

Using Lemma \ref{lem1} we see that the acceptance probability of MH algorithms \eqref{MH0} and \eqref{MH1} are identical and hence it is sufficient to analyse the convergence properties of \eqref{MH1} to determine the convergence properties of \eqref{MH0}.  

We will use the following Lyapunov central limit theorem, see e.g. \cite[Thm. 27.3]{billingsley1995}.

\begin{theorem}
\label{thm clt}
For each $d \in \mathbb{N}$ let $X_{d,1},\dotsc,X_{d,d}$ be a sequence of independent random variables each with finite expected value $\mu_{d,i}$ and variance $\sigma_{d,i}^2$.  Define
$
	s_d^2 := \sum_{i=1}^d \sigma_{d,i}^2.
$
If there exists a $\delta > 0$ such that
$$
	\lim_{d \rightarrow \infty} \frac{1}{s_d^{2+\delta}} \sum_{i=1}^d \mathrm{E}[ |X_{d,i} - \mu_{d,i}|^{2+\delta} ] = 0,
$$
then 
$$
	\frac{1}{s_d} \sum_{i=1}^d (X_{d,i} - \mu_{d,i}) \xrightarrow{\mathcal{D}} \normal(0,1) \qquad \mbox{as $d \rightarrow \infty$}.
$$
\end{theorem}

An equivalent conclusion to this theorem is $\sum_{i=1}^d X_{d,i} \rightarrow \normal(\sum_{i=1}^d \mu_{d,i}, s_d^2)$ in distribution as $d \rightarrow \infty$.  Another useful fact is  
\begin{equation}
\label{eq uf1}
	X \sim \normal(\mu,\sigma^2) \qquad \Rightarrow \qquad \mathrm{E}[ 1 \wedge \mathrm{e}^X ] = \Phi(\tsfrac{\mu}{\sigma}) + \mathrm{e}^{\mu + \sigma^2/2} \Phi(-\sigma - \tsfrac{\mu}{\sigma})
\end{equation}
where $\Phi$ is the standard normal cumulative distribution function.  See e.g. \cite[Prop. 2.4]{RGG1997} or \cite[Lem. B.2]{BRS2009}.

\begin{theorem}
\label{thm accept}
Suppose that $M$ and $N$ in \eqref{MH0} are functions of $A$, and the Markov chain is in equilibrium, i.e. $x \sim \normal(A^{-1}b,A^{-1})$.  
If there exists a $\delta > 0$ such that
\begin{equation}
\label{eq Tcond}
	\lim_{d \rightarrow \infty} \frac{\sum_{i=1}^d |T_{ji}|^{2+\delta}}{ \left( \sum_{i=1}^d |T_{ji}|^2 \right)^{1+\delta/2}} = 0 \qquad \mbox{for $j=1,2,3,4,5$}
\end{equation}
($j=0$ is not required) and the limits $\mu = \lim_{d\rightarrow \infty} \sum_{i=1}^d \mu_{d,i}$ and $\sigma^2 = \lim_{d\rightarrow \infty} \sum_{i=1}^d \sigma_{d,i}^2$ exist where
$$
	\mu_{d,i} = T_{0i} + T_{3i} + T_{4i} \qquad \mbox{and} \qquad \sigma_{d,i}^2 = T_{1i}^2 + T_{2i}^2 + 2T_{3i}^2 + 2T_{4i}^2 + T_{5i}^2,
$$
then 
$$
	Z:= \log \left( \frac{\pi(y)q(y,x)}{\pi(x)q(x,y)} \right) \xrightarrow{\mathcal{D}} \normal(\mu,\sigma^2) \qquad \mbox{as $d \rightarrow \infty$}
$$
and the expected acceptance probability satisfies
$$	
	\mathrm{E}[\alpha(x,y)] = \mathrm{E}[1 \wedge \mathrm{e}^Z] \rightarrow \Phi(\tsfrac{\mu}{\sigma}) + \mathrm{e}^{\mu + \sigma^2/2} \Phi(-\sigma - \tsfrac{\mu}{\sigma})
	\qquad \mbox{as $d \rightarrow \infty$.}
$$
\end{theorem}

In the above theorem, $M$ and $N$ may depend on $d$, so $T_{ji}$ may depend on $d$.

\begin{proof}
By Lemma \ref{lem2} it is sufficient to only consider \eqref{MH0} in the case where all matrices are diagonal matrices, e.g. $A = \operatorname{diag}(\lambda_1^2,\dotsc,\lambda_d^2)$, $\mathcal{A} = \operatorname{diag}(\tilde{\lambda}_1^2,\dotsc,\tilde{\lambda}_d^2)$, $M = \operatorname{diag}(M_1,\dotsc,M_d)$, $G = \operatorname{diag}(G_1,\dotsc,G_d)$, $m_i = \! \lambda_i^{-2} b_i$, and $\tilde{m}_i = \tilde{\lambda}_i^{-2} \beta_i$.  Then, in equilibrium we have 
$$
	x_i = m_i + \frac{1}{\lambda_i} \xi_i
$$
where $\xi_i \sim \normal(0,1)$ and using $\tilde{m}_i = G_i \tilde{m}_i + M_i^{-1}\beta_i$ we have
\begin{align*}
	y_i &= G_i x_i + \frac{\beta_i}{M_i} + \frac{(1-G_i^2)^{1/2}}{\tilde{\lambda}_i} \nu_i \\
	&= G_i \left( m_i + \frac{1}{\lambda_i} \xi_i \right) + (1-G_i) \tilde{m}_i  + \frac{g_i^{1/2}}{\tilde{\lambda}_i} \nu_i \\
	&= \tilde{m}_i + G_i \hat{r} + \frac{G_i}{\lambda_i} \xi_i + \frac{g_i^{1/2}}{\tilde{\lambda}_i} \nu_i 
\end{align*}
where $\nu_i \sim \normal(0,1)$.  From Lemma \ref{lem1} we also have $Z = \sum_{i=1}^d Z_{d,i}$ where
$$
	Z_{d,i} = -\tsfrac{1}{2} (\lambda_i^2 - \tilde{\lambda}_i^2)(y_i^2 - x_i^2) + (b_i-\beta_i)(y_i - x_i).
$$
Substituting $x_i$ and $y_i$ as above, using the identity $(b_i - \beta_i)\lambda_i^{-2} = \hat{r}_i + r_i \tilde{m}_i$, then after some algebra we eventually find
\begin{align*}
	Z_{d,i} &= T_{0i} + T_{1i}\xi_i + T_{2i}\nu_i + T_{3i} \xi_i^2 + T_{4i} \nu_i^2 + T_{5i} \xi_i \nu_i.
\end{align*}
Hence
$$
	\mu_{d,i}:= \mathrm{E}[Z_{d,i}] = T_{0i} + T_{3i} + T_{4i}
$$
and
\begin{align*}
	\sigma_{d,i}^2 &:= \mathrm{Var}[Z_{d,i}] = \mathrm{E}[Z_{d,i}^2] - \mathrm{E}[Z_{d,i}]^2 \\
	&= \left( T_{0i}^2 + T_{1i}^2 + T_{2i}^2 + 3 T_{3i}^2 + 3 T_{4i}^2 + T_{5i}^2 
	+ 2 T_{0i} T_{3i} + 2 T_{0i} T_{4i} + 2 T_{3i} T_{4i} \right) \\
	& \qquad - \left( T_{0i} + T_{3i} + T_{4i} \right)^2 \\
	&= T_{1i}^2 + T_{2i}^2 + 2 T_{3i}^2 + 2 T_{4i}^2 + T_{5i}^2	 
\end{align*}
and
\begin{align*}
	Z_{d,i} - \mu_{d,i} &= T_{1i} \xi_i + T_{2i} \nu_i + T_{3i} (\xi_i^2 - 1) + T_{4i} (\nu_i^2 - 1) + T_{5i} \xi_i \nu_i.
\end{align*}
Therefore, for any $d \in \mathbb{N}$ and $\delta > 0$ we can bound the Lyapunov condition in Theorem \ref{thm clt} as follows 
\begin{align*}
	\frac{1}{s_d^{2+\delta}} \sum_{i=1}^d \mathrm{E}[|Z_{d,i}-\mu_{d,i}|^{2+\delta}] 
	&\leq \frac{5^{2+\delta}}{s_d^{2+\delta}} \sum_{j=1}^5 C_j(\delta) \sum_{i=1}^d |T_{ji}|^{2+\delta} \\
	&\leq 5^{2+\delta} \sum_{j=1}^5 C_j(\delta) \frac{\sum_{i=1}^d |T_{ji}|^{2+\delta}}{\left( \sum_{i=1}^d T_{ji}^2 \right)^{1+\delta/2} }
\end{align*}
where $C_1(\delta) = C_2(\delta) = \mathrm{E}[|\xi|^{2+\delta}]$ , $C_3(\delta) = C_4(\delta) = \mathrm{E}[|\xi^2-1|^{2+\delta}]$
and $C_5(\delta) = \mathrm{E}[|\xi|^{2+\delta}]^2$, and $\xi \sim \normal(0,1)$.

Therefore, if \eqref{eq Tcond} holds then the result follows from Theorem \ref{thm clt} and \eqref{eq uf1}.
\end{proof}

\section{Expected squared jump size for a Gaussian target}
\label{sec jumpsize}

The efficiency of a MCMC method is usually given in terms of the integrated autocorrelation time which is equivalent to ``the number of dependent sample points from the Markov chain needed to give the variance reducing power of one independent point'', see e.g. \cite[\S 6.3]{N1993}.  Unfortunately, we are unable to directly estimate this quantity for our matrix splitting methods, and it depends on the statistic of concern.  As a proxy we instead consider the expected squared jump size of the Markov chain in a direction $q \in \bbR^d$,
$$
	\mathrm{E}[(q^T(x' - x))^2]
$$
where $x, x' \sim \normal(A^{-1}b,A^{-1})$ are successive elements of the Markov chain in equilibrium.  We will only consider the cases where $q$ is an eigenvector of the precision or covariance matrix.  It is related to the integrated autocorrelation time for the linear functional $q^T(\cdot)$ by
$$
	\mathrm{Corr}[q^T x,q^T x'] = 1 - \frac{\mathrm{E}[(q^T(x' - x))^2]}{2 \mathrm{Var}[q^T x]}
$$
so that large squared jump size implies small first-order autocorrelation, see e.g. \cite[\S 3]{RR1998} or \cite[\S 2.3]{BRS2009}.

This is similar to the approach used for analysing the efficiency of RWM, MALA and HMC, where the expected squared jump size of an arbitrary component of the Markov chain is considered, see e.g. \cite{BRS2009} and \cite{BPRSS2010}.  

We will need the following technical lemma whose proof is in the Appendix.

\begin{lemma}
\label{lem 5}
Suppose $\{ t_i \}_{i=1}^\infty \subset \bbR$ and $r > 0$.  Then, for any $k \in \mathbb{N}$,
\begin{equation}
\label{eq l1a}
	\lim_{d\rightarrow \infty} \frac{\sum_{i=1}^d |t_i|^r}{\left( \sum_{i=1}^d t_i^2 \right)^{r/2}} = 0 
	\quad \Rightarrow \quad
	\lim_{d\rightarrow \infty} \frac{\sum_{i=1, i \neq k}^d |t_i|^r}{\left( \sum_{i=1, i \neq k}^d t_i^2 \right)^{r/2}} = 0.
\end{equation}
\end{lemma}

The following theorem is a generalization of \cite[Prop. 3.8]{BPRSS2010}.

\begin{theorem}
\label{thm jumpsize}
Suppose that $M$ and $N$ in \eqref{MH0} are functions of $A$, and the Markov chain is in equilibrium, i.e. $x \sim \normal(A^{-1}b,A^{-1})$.
With $\mu_{d,i}$ and $\sigma_{d,i}^2$ defined as in Theorem \ref{thm accept}, let $q_i$ be a normalized eigenvector of $A$ corresponding to $\lambda_i^2$.  If there exists a $\delta > 0$ such that \eqref{eq Tcond} is satisfied, and $\mu^- := \lim_{d \rightarrow \infty} \sum_{j=1,j\neq i}^d \mu_{d,j}$ and $(\sigma^-)^2 := \lim_{d \rightarrow \infty} \sum_{j=1,j\neq i}^d \sigma_{d,j}^2$ exist, then 
\begin{equation}
\label{eq jump1}
	\mathrm{E}[(q_i^T (x'-x))^2]  = U_1 U_2 + E_3 + o(U_1) 
\end{equation}
as $d \rightarrow \infty$ where
\begin{align*}
	U_1 &= \tilde{g}_i^2 \hat{r}_i^2 + \frac{\tilde{g}_i^2}{\lambda_i^2} + \frac{g_i}{\tilde{\lambda}_i^2}, \\
	U_2 &= \mathrm{E}[1 \wedge \mathrm{e}^{X}] = \Phi(\tsfrac{\mu^-}{\sigma^-}) + \mathrm{e}^{\mu^- + (\sigma^-)^2/2} \Phi(-\sigma^- - \tsfrac{\mu^-}{\sigma^-}), \;\; X \sim N(\mu^-,(\sigma^-)^2),\\
	|E_3| &\leq U_3 = ( \sigma_{d,i}^2 + \mu_{d,i}^2 )^{1/2}  
	\times \left( \tilde{g}_i^4 \hat{r}_i^4 + \frac{3}{\lambda_i^4} (\tilde{g}_i^2 + \tilde{r}_i g_i)^2 + \frac{6}{\lambda_i^2} \hat{r}_i^2 \tilde{g}_i^2 ( \tilde{g}_i^2 + \tilde{r}_i g_i) \right)^{1/2}.
\end{align*}
\end{theorem}

\begin{proof}
Under the coordinate transformation $x \leftrightarrow Q^T x$, \eqref{MH0} becomes \eqref{MH1} and $\mathrm{E}[(q_i^T(x'-x))^2]$ becomes
$\mathrm{E}[(x_i' - x_i)^2]$.  Therefore it is sufficient to only consider the squared jump size of an arbitrary coordinate of the Markov chain for the case when all matrices are diagonal matrices.  As in the proof of Theorem \ref{thm accept}, let $A = \operatorname{diag}(\lambda_1^2,\dotsc,\lambda_d^2)$, $\mathcal{A} = \operatorname{diag}(\tilde{\lambda}_1^2,\dotsc,\tilde{\lambda}_d^2)$, $M = \operatorname{diag}(M_1,\dotsc,M_d)$, $G = \operatorname{diag}(G_1,\dotsc,G_d)$, $m_i = \lambda_i^{-2} b_i$, and $\tilde{m}_i = \tilde{\lambda}_i^{-2} \beta_i$.  Since the chain is in equilibrium we have $x_i = m_i + \lambda_i^{-1} \xi_i$ for $\xi_i \sim \normal(0,1)$ and $y_i = \tilde{m}_i + G_i \hat{r} + G_i \lambda_i^{-1} \xi_i + g_i^{1/2} \tilde{\lambda}_i^{-1} \nu_i$ where $\nu_i \sim \normal(0,1)$.  
Define $\alpha^-(x,y) := 1 \wedge \exp( \sum_{j=1,j\neq i}^d Z_{d,j} )$ where $Z_{d,j}$ is defined as in the proof of Theorem \ref{thm accept}.  

The proof strategy is now to approximate $\mathrm{E}[(x_i' - x_i)^2] = \mathrm{E}[(y_i - x_i)^2\alpha(x,y)]$ with $\mathrm{E}[(y_i-x_i)^2 \alpha^-(x,y)]$;
$$
	\mathrm{E}[(x_i'-x_i)^2] = \mathrm{E}[(y_i-x_i)^2 \alpha^-(x,y)] + \mathrm{E}[(\alpha(x,y)-\alpha^-(x,y))(y_i-x_i)^2].
$$
By independence, 
\begin{align*}
	\mathrm{E}[(y_i-x_i)^2 \alpha^-(x,y)] 
	&= \mathrm{E}[(y_i-x_i)^2] \mathrm{E}[\alpha^-(x,y)] \\
	&= \mathrm{E}\left[ \left(-\tilde{g}_i \hat{r}_i - \frac{\tilde{g}_i}{\lambda_i}\xi_i + \frac{g^{1/2}}{\tilde{\lambda}_i} \nu_i \right)^2 \right] \mathrm{E}[\alpha^-(x,y)]\\
	&= U_1 \mathrm{E}[\alpha^-(x,y)].
\end{align*}
Also, by Theorem \ref{thm accept} (using Lemma \ref{lem 5} to ensure the appropriate condition for Theorem \ref{thm accept} is met) we obtain $\mathrm{E}[\alpha^-(x,y)] \rightarrow U_2$ as $d \rightarrow \infty$, so $\mathrm{E}[(y_i-x_i)^2 \alpha(x,y)] = U_1 U_2 + \mathrm{o}(U_1)$ as $d\rightarrow \infty$.

The error is bounded using the Cauchy-Schwarz inequality;
$$
	|\mathrm{E}[ (\alpha(x,y) - \alpha^-(x,y))(y_i - x_i)^2]|
	\leq \mathrm{E}[(\alpha(x,y)-\alpha^-(x,y))^2]^{1/2} \mathrm{E}[(y_i-x_i)^4]^{1/2}.
$$
Since $1\wedge \mathrm{e}^X$ is Lipschitz with constant $1$, and using results from the proof of Theorem \ref{thm accept}, we obtain
$$
	\mathrm{E}[(\alpha(x,y)-\alpha^-(x,y))^2]^{1/2} 
	\leq \mathrm{E}[Z_{d,i}^2]^{1/2} 
	= ( \sigma_{d,i}^2 + \mu_{d,i}^2 )^{1/2},
$$
and some algebra yields
\begin{align*}
	\mathrm{E}[(y_i-x_i)^4]^{1/2}
	&= \mathrm{E}\left[ \left( -\tilde{g}_i \hat{r}_i - \frac{\tilde{g}_i}{\lambda_i}\xi_i + \frac{g^{1/2}}{\tilde{\lambda}_i} \nu_i \right)^4 \right]^{1/2} \\
	&= \left( \tilde{g}_i^4 \hat{r}_i^4 + \frac{3}{\lambda_i^4} (\tilde{g}_i^2 + \tilde{r}_i g_i)^2 + \frac{6}{\lambda_i^2} \hat{r}_i^2 \tilde{g}_i^2 ( \tilde{g}_i^2 + \tilde{r}_i g_i) \right)^{1/2}.
\end{align*}
\end{proof}

The terms in the theorem above are quite lengthy, but in many situations they simplify.  For example, we may have the situation where $\mu^- \rightarrow \mu$ and $\sigma^- \rightarrow \sigma$ as $d \rightarrow \infty$, so $U_2$ becomes the expected acceptance rate for the algorithm.  Also, it may be possible to derive a bound such as
$$
	|U_3| \leq C (r_i + \hat{r}_i)
$$	
so that $U_3$ is small if both the relative error of the $i^{\mathrm{th}}$ eigenvalue and error of the means are small.

\chapter{Non-Gaussian targets}
\label{sec nongaussian}

Our results can be extended to non-Gaussian target distributions in 
some cases.  We follow the methodology in \cite{BRS2009}.  Suppose that target $\pi_d$ is a change of measure from a reference Gaussian, so that $\pi_d$ is defined by \eqref{eq target0} and \eqref{eq ref}. 

With an AR(1) proposal associated with a matrix splitting $\mathcal{A} = M-N$, we consider an MH algorithm defined by
\begin{equation}
\label{MH0ng}
\begin{split}
	\mbox{Target:} & \qquad \pi_d, \\
	\mbox{Proposal:} & \qquad y = G x + g + \nu, \qquad \mbox{where $\nu \sim \normal(0,\Sigma)$}, 
\end{split}
\end{equation}
where $G = M^{-1}N$, $g = M^{-1} \beta$ and $\Sigma = M^{-1}(M^T+N)M^{-T} = \mathcal{A}^{-1} - G \mathcal{A}^{-1} G^T$.  The acceptance probability of this MH algorithm satisfies
$$
	\alpha(x,y) = 1 \wedge \exp\left( \phi_d(x)-\phi_d(y) + Z \right)
$$
where $Z = \log(\tsfrac{\tilde{\pi}_d(y) q(y,x)}{\tilde{\pi}_d(x) q(x,y)} )$.  Define $\tilde{\alpha}(x,y) = 1 \wedge \exp(Z)$ to be the acceptance probability for MH algorithm \eqref{MH0}.

In the theory below, $\mathrm{E}_{\pi_d}[\alpha(x,y)]$ is the expectation of $\alpha(x,y)$ over $x \sim \pi_d$ and $y$ from \eqref{MH0ng}.

\section{Expected acceptance rate for a non-Gaussian target}
\label{sec:acc rate ng}

The following theorem applies to inverse problems with a Gaussian prior and bounded likelihood.

\begin{theorem}
Suppose there exists a constant $M>0$ such that for sufficiently large $d$
$$
	|\phi_d(x)| \leq M \qquad \mbox{for all $x \in \bbR^d$}.
$$
Then MH algorithm \eqref{MH0ng} in equilibrium satisfies
\begin{align}
	\mathrm{E}_{\pi_d}[\alpha(x,y)] &\leq C \mathrm{E}_{\tilde{\pi}_d}[\tilde{\alpha}(x,y)] \label{ng1}\\
	\mathrm{E}_{\pi_d}[\alpha(x,y)] &> 0 \qquad \mbox{if $\mathrm{E}_{\tilde{\pi}_d}[|Z|] < \infty$}, \label{ng2}
\end{align}
where $\mathrm{E}_{\tilde{\pi}_d}[\tilde{\alpha}(x,y)]$ is the expected acceptance rate of \eqref{MH0} in equilibrium.
\end{theorem}

\begin{proof}  We follow the same reasoning as in the proof of \cite[Thm. 2]{BRS2009}.
Note that $1 \wedge \exp( \phi_d(x)-\phi_d(y) + Z) \leq \exp(2M) (1 \wedge \exp(Z))$, and $\pi_d(x) \leq \exp(M) \tilde{\pi}_d(x)$.  Hence, we obtain \eqref{ng1} with $C = \exp(3M)$.

To prove \eqref{ng2} first note for a random variable $X$ and any $\gamma > 0$ we have $\mathrm{E}[1 \wedge \exp(X)] \geq \exp(-\gamma)(1-\gamma^{-1} \mathrm{E}[|X|])$, see \cite[Lem. B1]{BRS2009}.  Also note that $C_0 := \mathrm{E}_{\pi_d}[|\phi_d(x)-\phi_d(y)+Z|] \leq C + C \mathrm{E}_{\tilde{\pi}_d}[|Z|] < \infty$.  Hence, we obtain \eqref{ng2} by taking $\gamma = 2 C_0$.
\end{proof}

Thus, in a certain weak sense, the acceptance rate of \eqref{MH0ng} with non-Gaussian target mimics the acceptance rate of \eqref{MH0} with a Gaussian target; if the acceptance rate of \eqref{MH0} is small, then so is the acceptance rate of \eqref{MH0ng}; and if the expected value of $|Z|$ is finite (which loosely corresponds to when the acceptance rate of \eqref{MH0} is positive) then the acceptance rate of \eqref{MH0ng} in equilibrium is positive.  

A similar result is given in \cite[Thm. 2]{BRS2009} for RWM and SLA.

Our next theorem more precisely describes the acceptance rate for a non-Gaussian target, but first, some definitions and a lemma.

We associate a norm with the spd precision matrix $A$ of our reference Gaussian measure $\tilde{\pi}$.  For any $s \in \bbR$ define a norm $| \cdot |_s$ on $\bbR^d$ by
$$
	| x |_s = | A^s x |
$$
for all $x \in \bbR^d$.  If $\lambda_1^2$ is the smallest eigenvalue of $A$, then
\begin{equation}
\label{eq:normbound}
	|x|_s \leq \lambda_1^{2(s-r)} |x|_r \qquad \mbox{for all $s<r$}.
\end{equation}

\begin{assumption}
\label{assumption1}
Suppose there exist constants $m,s,s',s'' \in \bbR$, $C,p>0$, and a locally bounded function $\delta: \bbR^+ \times \bbR^+ \mapsto \bbR^+$ such that for all sufficiently large $d$ 
\begin{align*}
	\phi_d(x) &\geq m, \\
	|\phi_d(x) - \phi_d(y)| &\leq \delta( |x-A^{-1}b|_s,|y-A^{-1}b|_s) |x-y|_{s'}, \\
	| \phi_d(x) | &\leq C (1 + |x-A^{-1}b|_{s''}^p ) 	
\end{align*}
for all $x,y \in \bbR^d$.
\end{assumption}

\begin{assumption}
\label{assumption2}
Suppose that $r$ is such that 
$$
	\lim_{d \rightarrow \infty} \sum_{i=1}^d \lambda_i^{4r-1} < \infty.
$$
\end{assumption}

The following lemma is similar to part of the proof of \cite[Thm. 3]{BRS2009}.

\begin{lemma}
\label{lem:phi}
Suppose $\phi_d$ satisfies Assumption \ref{assumption1} and $r = \max \{ s,s',s'' \}$ satisfies Assumption \ref{assumption2}.  Also suppose that there exists $t_{d,i}$ and a $t > 0$ such that MH algorithm \eqref{MH0ng} satisfies
\begin{align*}
	&\mbox{$G$ and $\Sigma$ are functions of $A$}, \\
	&\tilde{g}_i^2\hat{r}_i^2 \lambda_i, \; \tilde{g}_i^2 \lambda_i^{-1}, \; g_i \tilde{\lambda}_i^{-1} \mbox{ are $\mathcal{O}(t_{d,i}) = \mathcal{O}(d^{-t})$ (uniformly in $i$) as $d \rightarrow \infty$, and } \\
	&\tilde{r}_i \mbox{ is bounded uniformly in $d$ and $i$}.
\end{align*}
Then for any $q \in \mathbb{N}$ there exists a constant $C > 0$ (that may depend on $q$) such that
\begin{align*}
	&\mathrm{E}_{\tilde{\pi}_d}[|x-A^{-1}b|_r^{2q}] < C \mbox{ for all $d$} \\
	&\mathrm{E}_{\tilde{\pi}_d}[|y-x|_r^{2q}] = \mathcal{O}(t_{d,i}^{q})= \mathcal{O}(d^{-qt}) \mbox{ as $d \rightarrow \infty$},
\end{align*}
and for proposal $y$ from $x$,
$$
	\phi_d(x) - \phi_d(y) \rightarrow 0 \qquad \mbox{in $L^q(\pi_d)$ as $d \rightarrow \infty$}.
$$
\end{lemma}

\begin{proof}
For $x \sim \normal(A^{-1}b,A^{-1})$, $\xi \sim \normal(0,I)$, $\nu \sim \normal(0,I)$ and $\Lambda = \operatorname{diag}(\lambda_1^2,\dotsc,\lambda_d^2)$, 
\begin{align*}
	\mathrm{E}_{\tilde{\pi}_d}[ |x - A^{-1}b |_r^{2q} ] 
	&= \mathrm{E}_{\tilde{\pi}_d}[ | A^{r-1/2} \xi |^{2q} ] 
	= \mathrm{E}_{\tilde{\pi}_d}[ | \Lambda^{r-1/2} \nu |^{2q} ] \\
	&= \mathrm{E}_{\tilde{\pi}_d}\left[ \left( \sum_{i=1}^d \lambda_i^{4r-1} \nu_i^2 \right)^{q} \right] 
	\leq C \left( \sum_{i=1}^d \lambda_i^{4r-1} \right)^q,
\end{align*}
which is bounded uniformly in $d$ by Assumption \ref{assumption2}.

As above, and using the transformation $x \leftrightarrow Q^T x$, and $(y_i - x_i) = -\tilde{g}_i \hat{r}_i - \tilde{g}_i \lambda_i^{-1} \xi_i + g_i^{1/2} \tilde{\lambda}_i^{-1} \nu_i$ from the proof of Theorem \ref{thm jumpsize} where $\xi_i$ and $\nu_i \stackrel{\text{iid}}{\sim} \normal(0,1)$, we have
\begin{align*}
	\mathrm{E}_{\tilde{\pi}_d}[ |y - x|_r^{2q} ]
	&= \mathrm{E}_{\tilde{\pi}_d}\left[ \left( \sum_{i=1}^d \lambda_i^{4r} \left( -\tilde{g}_i \hat{r}_i - \frac{\tilde{g}_i}{\lambda_i} \xi_i + \frac{g_i^{1/2}}{\tilde{\lambda}_i} \nu_i \right)^2 \right)^{q} \right]  \\
	&= \mathrm{E}_{\tilde{\pi}_d}\left[ \left( \sum_{i=1}^d \lambda_i^{4r-1} \left( -\tilde{g}_i \hat{r}_i \lambda_i^{1/2} - \frac{\tilde{g}_i}{\lambda_i^{1/2}} \xi_i + \frac{\tilde{r}_i^{1/4} g_i^{1/2}}{\tilde{\lambda}_i^{1/2}}  \nu_i \right)^2 \right)^{q} \right]  .
\end{align*}
Since $\tilde{g}_i^2 \hat{r}_i^2 \lambda_i$, $\tilde{g}_i^2 \lambda^{-1}$, $g_i \tilde{\lambda}_i^{-1} = \mathcal{O}(t_{d,i})$ uniformly in $i$, and $\tilde{r}_i$ is bounded uniformly in $d$ and $i$, it follows that 
for all sufficiently large $d$, 
$$
	\mathrm{E}_{\tilde{\pi}_d}[ |y - x|_r^{2q} ] \leq C t_{d,i}^{q} \left( \sum_{i=1}^d \lambda^{2r-1} \right)^q,
$$
so by Assumption \ref{assumption2}, $\mathrm{E}_{\tilde{\pi}_d}[ |y - x|_r^{2q} ] = \mathcal{O}(t_{d,i}^{q})$ as $d \rightarrow \infty$.

From $\mathrm{E}_{\tilde{\pi}_d}[|x-A^{-1}b|_r^{2q}] < C$ and $\mathrm{E}_{\tilde{\pi}_d}[|y-x|_r^{2q}] \rightarrow 0$, it follows from the triangle inequality that there is a (new) constant $C>0$ such that
\begin{equation}
\label{eq:ybd}
	\mathrm{E}_{\tilde{\pi}_d}[|y-A^{-1}b|_r^{2q}] < C \qquad \mbox{for all $d$}.
\end{equation}

Let $\Delta_d = \phi_d(x) - \phi_d(y)$.  For any $R>0$ define
\begin{align*}
	\gamma(R) &= \sup \{  \delta(a,b)^q : a\leq R,b\leq R \} \\
	S_1 &= \{ x \in \bbR^d : |x-A^{-1}b|_s \leq R \}, \\
	S_2 &= \{ x \in \bbR^d: |y-A^{-1}b|_s \leq R \},
\end{align*}
and let $\mathbb{I}_{S}$ be the indicator function for set $S$.  Using Assumption \ref{assumption1}, a generic constant $C$ that may vary between lines, the Cauchy-Schwarz inequality, and then Markov's inequality, we have for each $q \in \mathbb{N}$
\begin{align*}
	\mathrm{E}_{\tilde{\pi}_d}[|\Delta_d|^q] 
	 &= \mathrm{E}_{\tilde{\pi}_d}[|\Delta_d|^q \mathbb{I}_{S_1 \cap S_2} ] +  \mathrm{E}_{\tilde{\pi}_d}[|\Delta_d|^q \mathbb{I}_{\bbR^d \backslash (S_1 \cap S_2)}] \\
	&\leq \gamma(R) \mathrm{E}[ |x-y|_{s'}^q ] + C \mathrm{E}[(1+|x-A^{-1}b|_{s''}^{pq} + |y-A^{-1}b|_{s''}^{pq}) \mathbb{I}_{\bbR^d \backslash (S_1 \cap S_2)}] \\
	&\leq \gamma(R) \mathrm{E}[ |x-y|_{s'}^q ] \\
	&\qquad + C \mathrm{E}[1+|x-A^{-1}b|_{s''}^{2pq} + |y-A^{-1}b|_{s''}^{2pq}]^{1/2} (\mathrm{P}(\bbR^d\backslash S_1) + \mathrm{P}(\bbR^d \backslash S_2))^{1/2} \\
	&\leq \gamma(R) \mathrm{E}[ |x-y|_{s'}^q ] + C (\mathrm{P}(\bbR^d\backslash S_1) + \mathrm{P}(\bbR^d \backslash S_2))^{1/2} \\
	&\leq \gamma(R) \mathrm{E}[ |x-y|_{s'}^q ] + \frac{C}{R^{1/2}}  \left( \mathrm{E}[|x-A^{-1}b|_s] + \mathrm{E}[|y-A^{-1}b|_s] \right)^{1/2}\\
	&\leq \gamma(R) \mathrm{E}[ |x-y|_{s'}^q ] + \frac{C}{R^{1/2}}.
\end{align*}
Note that we used Jensen's inquality (which implies $\nrm{f}_{L^p(\tilde{\pi}_d)} \leq \nrm{f}_{L^q(\tilde{\pi}_d)}$ for $1\leq p \leq q \leq \infty$) and \eqref{eq:normbound} to obtain bounds on $\mathrm{E}[|x-A^{-1}b|_{s''}^{2pq}]$, $\mathrm{E}[|y-A^{-1}b|_{s''}^{2pq}]$, $\mathrm{E}[|x-A^{-1}b|_s]$ and $\mathrm{E}[|y-A^{-1}b|_s]$.

Hence, for any $\epsilon > 0$ we can choose $R = R(\epsilon)$ such that $C/R^{1/2} < \epsilon/2$ and since $\mathrm{E}_{\tilde{\pi}_d}[|y-x|_r^{q}] \rightarrow 0$ as $d \rightarrow \infty$ (by Jensen's inequality), for all sufficiently large $d$ we have
$$ 
\mathrm{E}_{\tilde{\pi}_d}[|\phi_d(x) - \phi_d(y)|] = \mathrm{E}_{\tilde{\pi}_d}[|\Delta_d|^q] < \epsilon.
$$
Thus, 
$$
	\phi_d(x)-\phi_d(y) \rightarrow 0 \qquad \mbox{ in $L^q(\tilde{\pi}_d)$ as $d \rightarrow \infty$}.
$$
The result then follows from $\phi_d(x) \geq m$.
\end{proof}

\begin{theorem}
\label{thm:ng1}
Suppose that $\phi_d$ satisfies Assumption \ref{assumption1} and $r = \max \{ s,s',s'' \}$ satisfies Assumption \ref{assumption2}.  Also suppose there exists $t_{d,i}$ and $t>0$ such that MH algorithm \eqref{MH0ng} satisfies
\begin{align*}
	&\mbox{$G$ and $\Sigma$ are functions of $A$}, \\
	&\tilde{g}_i^2\hat{r}_i^2 \lambda_i, \; \tilde{g}_i^2 \lambda_i^{-1}, \; g_i \tilde{\lambda}_i^{-1} \mbox{ are $\mathcal{O}(t_{d,i}) = \mathcal{O}(d^{-t})$ (uniformly in $i$) as $d \rightarrow \infty$,} \\
	&\tilde{r}_i \mbox{ is bounded uniformly in $d$ and $i$}, \\
	&T_{1i}, T_{2i}, T_{3i}, T_{4i}, T_{5i} \mbox{ are } \mathcal{O}(d^{-1/2}) \mbox{ as $d \rightarrow \infty$ (uniformly in $i$), and} \\
	&\lim_{d\rightarrow\infty} \sum_{i=1}^d T_{1i}^2 + T_{2i}^2 + 2 T_{3i}^2 + 2 T_{4i}^2 + T_{5i}^2 < \infty.
\end{align*}
If 
\begin{align*}
	\mu_{ng} &= \mu + \lim_{d \rightarrow \infty} \sum_{i=1}^d \kappa_{d,i} T_{1i} + T_{3i}(\gamma_i-1),  \\
	\sigma_{ng}^2 &= \sigma^2 + \lim_{d\rightarrow \infty} \sum_{i=1}^d  \left(\kappa_{d,i} T_{1i} + T_{3i} (\gamma_{d,i} - 1) \right)^2,
\end{align*}
exist, where $\kappa_{d,i} = \mathrm{E}_{\pi_d}[q_i^T A^{1/2}(x - A^{-1}b)]$, $\gamma_{d,i} = \mathrm{E}_{\pi_d}[(q_i^T A^{1/2}(x - A^{-1}b))^2]$, $q_i$ is a normalised eigenvector of $A$ corresponding to the eigenvalue $\lambda_i^2$, and $\mu$ and $\sigma^2$ are the same as in Theorem \ref{thm accept}, then
$$
	\mathrm{E}_{\pi_d}[\alpha(x,y)] \rightarrow \mathrm{E}[1 \wedge \mathrm{e}^{Z_{ng}}] = \Phi(\tsfrac{\mu_{ng}}{\sigma_{ng}}) + \mathrm{e}^{\mu_{ng} + \sigma_{ng}^2/2} \Phi(-\sigma_{ng} - \tsfrac{\mu_{ng}}{\sigma_{ng}})
$$
as $d \rightarrow \infty$, where $Z_{ng} \sim \normal(\mu_{ng},\sigma_{ng}^2)$.
\end{theorem}

\begin{proof}
As in the proofs of Theorems \ref{thm accept} and \ref{thm jumpsize} it is sufficient to prove the result in the case when all matrices are diagonal.  This follows from the coordinate transformation $z = Q^T x$ where $A = Q \Lambda Q^T$, since
$$
	|x|_s = |\Lambda^s z| \qquad \mbox{and} \qquad
	|x - A^{-1}b|_s = |\Lambda^s( z - \Lambda^{-1}Q^Tb)|
$$
for all $x \in \bbR^d$ and $s \in \bbR$, and since $\psi_d(z) := \phi_d(Q z)$ satisfies Assumption \ref{assumption1}.  Henceforth and without loss of generality, let us assume that $A$, $G$ and $\Sigma$ are diagonal matrices.


Using Lemma \ref{lem:phi} with $q = 1$, and the fact that $z \mapsto 1 \wedge \exp(z)$ is globally Lipschitz continuous, it follows that
\begin{equation}
\label{eq:acceptlim}
	\mathrm{E}_{\pi_d}[\alpha(x,y)] - \mathrm{E}_{\pi_d}[\tilde{\alpha}(x,y)] \rightarrow 0 \qquad \mbox{as $d \rightarrow \infty$}.
\end{equation}
To complete the proof we must find the limit of $\mathrm{E}_{\pi_d}[\tilde{\alpha}(x,y)]$ as $d \rightarrow \infty$.

As in the proof of Theorem \ref{thm accept} we have $Z = \sum_{i=1}^d Z_{d,i}$ where
$$
	Z_{d,i} = T_{0i} + T_{1i} \xi_i + T_{2i} \nu_i + T_{3i} \xi_i^2 + T_{4i} \nu_i^2 + T_{5i} \xi_i \nu_i,
$$
noting that $\xi_i = \lambda_i (x_i - m_i)$ with $x \sim \pi_d$, $m_i = (A^{-1}b)_i$, and $\nu_i \sim \normal(0,1)$.  

Note that $\kappa_{d,i} = \mathrm{E}_{\pi_d}[ \lambda_i (x_i - m_i)] = \mathrm{E}_{\pi_d}[\xi_i]$ and $\gamma_{d,i} = \mathrm{E}_{\pi_d}[ \lambda_i^2 (x_i - m_i)^2] = \mathrm{E}_{\pi_d}[\xi_i^2]$ are both uniformly bounded in $d$ and $i$ since $\phi_d(x) \geq m$ and $|\kappa_{d,i}|\leq \mathrm{e}^{-m} \mathrm{E}_{\tilde{\pi}_d}[|\xi_i|] = \mathrm{e}^{-m} \mathrm{E}[|u|]$ and $0 \leq \gamma_{d,i} \leq \mathrm{e}^{-m} \mathrm{E}_{\tilde{\pi}_d}[\xi_i^2] = \mathrm{e}^{-m} \mathrm{E}[u^2]$ where $u \sim \normal(0,1)$.  If we define 
$$
	S_{d,j} := \sum_{i=1}^j T_{1i} (\xi_i - \kappa_{d,i}) + T_{2i} \nu_i + T_{3i} (\xi_i^2 - \gamma_{d,i}) + T_{4i} (\nu_i^2 - 1) + T_{5i} \xi_i \nu_i,
$$
then
$$
	Z = \sum_{i=1}^d (T_{0i} + T_{1i} \kappa_{d,i} + T_{3i} \gamma_{d,i} + T_{4i}) + S_{d,d}.
$$

We will now show that $S_{d,d}$ converges in distribution towards a normal distribution as $d \rightarrow \infty$, using a Martingale central limit theorem, see \cite[Thm. 3.2, p. 58]{HallHeyde}.  

The set $\{ S_{d,j} : 1 \leq j \leq d, d \in \mathbb{N} \}$ is a zero mean, square-integrable Martingale array, i.e. for each $d \in \mathbb{N}$ and $1 \leq j \leq d$, $S_{d,j}$ is measurable, 
$$
	\mathrm{E}_{\pi_d}[S_{d,j}] = 0, \qquad
	\mathrm{E}_{\pi_d}[|S_{d,j}|]<\infty, \qquad \mbox{and} \qquad
	\mathrm{E}_{\pi_d}[(S_{j,d})^2] < \infty.
$$
For definitions, see \cite[p. 1 and 53]{HallHeyde}. 
Define $X_{d,j} := S_{d,j} - S_{d,j-1}$.  To ensure we satisfy the conditions for \cite[Thm. 3.2]{HallHeyde} we must show that there exists an a.s. finite random variable $\eta^2$ such that
\begin{align}
	&\max_{1 \leq i \leq d} | X_{d,i} | \xrightarrow{p} 0 \qquad \mbox{as $d \rightarrow \infty$}, \label{eq:c1}\\
	&\sum_{i=1}^d X_{d,i}^2 \xrightarrow{p} \eta^2 \qquad \mbox{as $d \rightarrow \infty$, and} \label{eq:c2} \\
	&\mathrm{E}_{\pi_d} \left( \max_{1\leq i \leq d} X_{d,i}^2 \right) \mbox{ is bounded in $d$}. \label{eq:c3}
\end{align}

First consider \eqref{eq:c1}.  We have
$$
	|X_{d,i}| \leq |T_{1i}|(|\xi_i| + |\kappa_{d,i}|) + |T_{2i}| |\nu_i| + |T_{3i}| (\xi_i^2 + \gamma_{d,i}) + |T_{4i}| (\nu_i^2 + 1) + |T_{5i}| |\xi_i| |\nu_i|
$$
which goes to zero in probability since $\kappa_{d,i}$ and $\gamma_{d,i}$ are bounded uniformly and $|T_{ji}|$ are all $\mathcal{O}(d^{-1/2})$ as $d \rightarrow \infty$ uniformly in $i$.

Now consider \eqref{eq:c2}.  Define 
$$
	\eta^2 := \sigma_{ng}^2 = \lim_{d \rightarrow \infty} \sum_{i=1}^d T_{1i}^2 + T_{2i}^2 + 2 T_{3i}^2 + 2 T_{4i}^2 + T_{5i}^2 + \left( \kappa_{d,i} T_{1i} + T_{3i} ( \gamma_{d,i} - 1 ) \right)^2 < \infty
$$
and $Y_{d,i} := d X_{d,i}^2$ so that $\overline{Y}_{d} = \frac{1}{d} \sum_{i=1}^d Y_{d,i} = \sum_{i=1}^d X_{d,i}^2$.  Then
$$
	\mathrm{E}_{\tilde{\pi}_d}[Y_{d,i}] = d \left( T_{1i}^2 + T_{2i}^2 + 2 T_{3i}^2 + 2 T_{4i}^2 + T_{5i}^2 + (\kappa_{d,i}T_{1i} + T_{3i} (\gamma_{d,i} - 1))^2 \right),
$$
and
$$
	\mathrm{Var}_{\tilde{\pi}_d}[Y_{d,i}] = d^2 \sum_{|\omega| = 4} C_{\omega} T_{1i}^{\omega_1} T_{2i}^{\omega_2} T_{3i}^{\omega_3} T_{4i}^{\omega_4} T_{5i}^{\omega_5}
$$
where $\omega$ is a multi-index with $\omega_i \geq 0$ and $|\omega| = \sum_{i} \omega_i = 4$, and $C_\omega$ are uniformly bounded constants.  Since $T_{ji}$ are all $\mathcal{O}(d^{-1/2})$, $\mathrm{E}_{\tilde{\pi}_d}[Y_{d,i}]$ and $\mathrm{Var}_{\tilde{\pi}_d}[Y_i]$ are uniformly bounded.

Then by the Markov inequality, and independence of $Y_{d,i}$, for any $\epsilon > 0$,
\begin{align*}
	\mathrm{Pr}\left( \left| \overline{Y}_d - \mathrm{E}_{\tilde{\pi}_d}\left[ \overline{Y}_d \right] \right| \geq \epsilon \right) 
	&\leq \frac{1}{\epsilon^2} \mathrm{Var}_{\tilde{\pi}_d}\left[ \overline{Y}_d \right] \\
	&= \frac{1}{\epsilon^2 d^2} \sum_{i=1}^d \mathrm{Var}_{\tilde{\pi}_d} [Y_{d,i} ] \\
	&\leq \frac{C}{\epsilon^2 d} \rightarrow 0 \quad \mbox{as $d \rightarrow \infty$}.
\end{align*}
Hence $\overline{Y}_d \xrightarrow{p} \eta^2$ as $d \rightarrow \infty$.  This is not \eqref{eq:c2} yet, because it is convergence with respect to $\tilde{\pi}_d$ rather than $\pi_d$.  

Since $\lim_{d \rightarrow \infty} \mathrm{E}_{\tilde{\pi}_d}[ \overline{Y}_d ] = \eta^2 < \infty$ and $|\overline{Y}_d| = \overline{Y}_d$, it follows that $\mathrm{E}_{\tilde{\pi}_d}[|\overline{Y}_d| ]$ is uniformly bounded in $d$.  Therefore, $\overline{Y}_d$ is uniformly integrable and so $\overline{Y}_d \rightarrow \eta^2$ in $L^1(\tilde{\pi}_d)$ as $d \rightarrow \infty$ \cite[Thm. 6.5.5 on p. 169]{sensinger1993}.  Hence 
$$
	\sum_{i=1}^d X_{d,i}^2 \xrightarrow{L^1(\tilde{\pi}_d)} \eta^2 \qquad \mbox{as $d \rightarrow \infty$}.
$$
From $\phi_d(x) \geq m$, the same limit holds in $L^1(\pi_d)$, which also implies convergence in probability, hence we have shown \eqref{eq:c2}.

Condition \eqref{eq:c3} follows from $X_{d,i}^2 \leq \overline{Y}_d$ for $1 \leq i \leq d$, $\mathrm{E}_{\tilde{\pi}_d}[|\overline{Y}_d| ]$ uniformly bounded in $d$, and $\phi_d(x) \geq m$.

Therefore, by the Martingale central limit theorem \cite[Thm. 3.2]{HallHeyde}, 
$$
	S_{d,d} \xrightarrow{\mathcal{D}} \normal(0,\eta^2) \qquad \mbox{as $d \rightarrow \infty$}.
$$	
Hence, 
\begin{equation}
\label{eq:zlim}
	Z  \xrightarrow{\mathcal{D}} \normal(\mu_{ng}, \sigma_{ng}^2 ) \qquad \mbox{as $d \rightarrow \infty$}.
\end{equation}
The result then follows from \eqref{eq:acceptlim}, \eqref{eq:zlim} and \eqref{eq uf1}.
\end{proof}

\begin{corollary}
\label{cor:ng}
In addition to the conditions for Theorem \ref{thm:ng1}, if
$$
	\lim_{d \rightarrow \infty} \sum_{i=1}^d T_{1i}^2 + T_{3i}^2 = 0
$$
then 
$$
	\mu_{ng} = \mu \qquad \mbox{and} \qquad \sigma_{ng}^2 = \sigma^2,
$$
and the expected acceptance rate for the non-Gaussian target case has the same limit as $d \rightarrow \infty$ as the Gaussian target case.
\end{corollary}

\begin{proof}
With $\xi_i$ defined as in the proof of Theorem \ref{thm:ng1} (we only need to consider the case when matrices are diagonal), since $\mathrm{E}_{\tilde{\pi}_d}[\xi_i^2] = 1$ and $\mathrm{E}_{\tilde{\pi}_d}[\xi_i^4] = 3$ for all $i$ and $d$, and since $\xi_i$ and $\xi_j$ are independent for $i \neq j$ under $\tilde{\pi}_d$,
$$
	\lim_{d \rightarrow \infty} \mathrm{E}_{\tilde{\pi}_d}\left[ \left( \sum_{i=1}^d T_{1i} \xi_i + T_{3i} (\xi_i^2 - 1) \right)^2 \right] = \lim_{d \rightarrow \infty} \sum_{i=1}^d T_{1i}^2 + 2T_{3i}^2 = 0
$$
From Jensen's inequality and $\phi_d(x) \geq m$ we have $\sum_{i=1}^d (T_{1i} \xi_i + T_{3i} (\xi_i^2 - 1)) \rightarrow 0$ in $L^1(\pi_d)$ as $d \rightarrow \infty$.  Therefore, since $\kappa_{d,i} = \mathrm{E}_{\pi_d}[\xi_i]$ and $\gamma_{d,i} = \mathrm{E}_{\pi_d}[\xi_i^2]$,
$$
	\sum_{i=1}^d T_{1i} \kappa_{d,i} + T_{3i}(\gamma_{d,i} - 1) = \mathrm{E}_{\pi_d}\left[ \sum_{i=1}^d T_{1i} \xi_i + T_{3i}(\xi_i^2 - 1) \right] \rightarrow 0 \qquad \mbox{as $d \rightarrow \infty$},
$$
and $\mu_{ng} = \mu$.  Also, since $\kappa_{d,i}$ and $\gamma_{d,i}$ are uniformly bounded in $d$ and $i$,
$$
	\sum_{i=1}^d (T_{1i}\kappa_{d,i} + T_{3i} (\gamma_{d,i}-1))^2 \leq C \sum_{i=1}^d T_{1i}^2 + T_{3i}^2 \rightarrow 0 \qquad \mbox{as $d \rightarrow \infty$},
$$
and $\sigma_{ng}^2 = \sigma^2$.
\end{proof}


\section{Expected jumpsize for non-Gaussian target}

\begin{theorem}
\label{thm:ng2}
Under the same conditions as Theorem \ref{thm:ng1}, 
\begin{multline*}
	\mathrm{E}_{\pi_d}[ (q_i^T (x' - x))^2 ] = \left( \left(\tilde{g}_i^2 \hat{r}_i^2 + \frac{g_i}{\tilde{\lambda}_i^2} \right) \mathrm{E}_{\pi_d}[\alpha(x,y)] + 2 \frac{\hat{r}_i \tilde{g}_i^2 \gamma_{d,i}^{1/2}}{\lambda_i} u_{d,i} + \frac{\tilde{g}_i^2 \gamma_{d,i}}{\lambda_i^2} v_{d,i} \right) + \mathrm{o}(t_{d,i} \lambda_i^{-1}) 
\end{multline*}
(uniformly in $i$) as $d \rightarrow \infty$, for some $-1 \leq u_{d,i} \leq 1$ and $0 \leq v_{d,i} \leq 1$.
\end{theorem}

\begin{proof}
As in earlier proofs, it is sufficient to prove the result in the case when all matrices are diagonal, so let $A$, $G$ and $\Sigma$ be diagonal matrices.  

Let $S_d$ denote the expected squared jump size in coordinate direction $i$, so that
$$
	S_d = \mathrm{E}_{\pi_d}[(x_i' - x_i)^2] = \mathrm{E}_{\pi_d}[(y_i - x_i)^2 \alpha(x,y)].
$$
Also define
\begin{align*}
	\tilde{S}_d := \mathrm{E}_{\pi_d}[(y_i-x_i)^2 \tilde{\alpha}(x,y)] \quad \mbox{and} \quad
	\tilde{S}_d^- := \mathrm{E}_{\pi_d}[(y_i-x_i)^2 \tilde{\alpha}^-(x,y)]
\end{align*}
where $\tilde{\alpha}^-(x,y) = 1 \wedge \exp( \sum_{j=1,j\neq i}^d Z_{d,i})$, and $Z_{d,i}$ is the same as earlier.  Recall that $\xi_i = \lambda_i (x_i-m_i)$.

We first show that $\mathrm{E}_{\pi_d}[(y_i-x_i)^4] = \mathcal{O}(t_{d,i}^2 \lambda_i^{-2})$ (uniformly in $i$) as $d \rightarrow \infty$.  As in the proof of Theorem \ref{thm jumpsize}, from $y = Gx + g + \Sigma^{1/2} \nu$ where $\nu \sim \normal(0,I)$, and $\tilde{m} = G \tilde{m} + g$ it follows that
$$
	y_i - x_i = -\tilde{g}_i \hat{r}_i - \frac{\tilde{g}_i}{\lambda_i} \xi_i + \frac{g_i^{1/2}}{\tilde{\lambda}_i} \nu_i
$$
where $\nu_i \stackrel{\text{iid}}{\sim} \normal(0,1)$, $\xi_i = \lambda_i (x_i - m_i)$ and $x \sim \pi_d$.  Therefore,
$$
	(y_i - x_i)^4 = \lambda_i^{-2} \left( -\tilde{g}_i \hat{r}_i \lambda_i^{1/2} - \frac{\tilde{g}_i}{\lambda_i^{1/2}} \xi_i + \frac{\tilde{r}_i^{1/4} g_i^{1/2}}{ \tilde{\lambda}_i^{1/2}}  \nu_i \right)^4.
$$
In the proof of Theorem \ref{thm:ng1} we showed that $|\kappa_{d,i}| = |\mathrm{E}_{\pi_d}[\xi_i]|$ and $\gamma_{d,i} = \mathrm{E}_{\pi_d}[\xi_i^2]$ are uniformly bounded.  Similarly, $|\mathrm{E}_{\pi_d}[\xi_i^3]|$ and $\mathrm{E}_{\pi_d}[\xi_i^4]$ are also uniformly bounded.  Using these facts and $\tilde{g}_i^2 \hat{r}_i^2 \lambda_i$, $\tilde{g}_i^2 \lambda^{-1}$, $g_i \tilde{\lambda}_i^{-1}$ $\mathcal(t_{d,i})$ (uniformly in $i$) and $\tilde{r}_i$ bounded uniformly in $d$ and $i$, it follows that $\mathrm{E}_{\pi_d}[(y_i-x_i)^4] = \mathcal{O}(t_{d,i}^2 \lambda_i^{-2})$ (uniformly in $i$) as $d \rightarrow \infty$.

Now let us show that $S_d - \tilde{S}_d = \mathrm{o}(t_{d,i}\lambda_i^{-1})$  as $d \rightarrow \infty$.  From the Lipschitz continuity of $z \mapsto 1 \wedge \exp(z)$ and the Cauchy-Schwarz inequality,
\begin{align*}
	S_d - \tilde{S}_d &= \mathrm{E}_{\pi_d}[ (y_i - x_i)^2 (\alpha(x,y) - \tilde{\alpha}(x,y) ] \\
	&\leq \mathrm{E}_{\pi_d}[(y_i-x_i)^4]^{1/2} \mathrm{E}_{\pi_d}[ |\phi_d(x) - \phi_d(y) |^2 ]^{1/2} = \mathrm{o}(t_{d,i} \lambda_i^{-1}) \quad \mbox{as $d \rightarrow \infty$},
\end{align*}
by Lemma \ref{lem:phi}, since $\mathrm{E}_{\pi_d}[(y_i-x_i)^4]  =\mathcal{O}(t_{d,i}^2\lambda_i^{-2})$.

Now show that $\tilde{S}_d - \tilde{S}_d^- = \mathrm{o}(t_{d,i} \lambda_i^{-1})$ as $d \rightarrow \infty$.  Again, by the Lipschitz continuity of $z \mapsto 1 \wedge \exp(z)$ and the Cauchy-Schwarz inequality,
\begin{align*}
	\tilde{S}_d - \tilde{S}_d^- &= \mathrm{E}_{\pi_d}[ (y_i - x_i)^2 (\tilde{\alpha}(x,y) - \tilde{\alpha}^-(x,y))] \\
	&\leq \mathrm{E}_{\pi_d}[(y_i-x_i)^4]^{1/2} \mathrm{E}_{\pi_d}[ Z_{d,i}^2 ]^{1/2} = \mathrm{o}(t_{d,i}\lambda_i^{-1}) \quad \mbox{as $d \rightarrow \infty$},
\end{align*}	
since $\mathrm{E}_{\pi_d}[(y_i-x_i)^4] = \mathcal{O}(t_{d,i}^2 \lambda_i^2)$ and $T_{ji}$ are all $\mathcal{O}(d^{-1/2})$.

Now consider $\tilde{S}_d^-$.  Since $\nu_i$ is independent of $\xi_i$ and $\tilde{\alpha}^-(x,y)$, $\mathrm{E}[\nu_i] = 0$ and $\mathrm{E}[\nu_i^2] = 1$,
\begin{align*}
	\tilde{S}_d^-
	&= \mathrm{E}_{\pi_d}[ (y_i - x_i)^2 \tilde{\alpha}^-(x,y)] \\
	&= \mathrm{E}_{\pi_d}\left[ \left( \tilde{g}_i^2 \hat{r}_i^2 + 2 \frac{\hat{r}_i \tilde{g}_i^2}{\lambda_i} \xi_i + \frac{\tilde{g}_i^2}{\lambda_i^2} \xi_i^2 + \frac{g_i}{\tilde{\lambda}_i^2} \right) \tilde{\alpha}^-(x,y) \right] \\
	&= \left( \tilde{g}_i^2 \hat{r}_i^2 + \frac{g_i}{\tilde{\lambda}_i^2} \right) \mathrm{E}_{\pi_d}[ \tilde{\alpha}^-(x,y)] + 2 \frac{\hat{r}_i \tilde{g}_i^2}{\lambda_i} \mathrm{E}_{\pi_d} [\xi_i \tilde{\alpha}^-(x,y) ] + \frac{\tilde{g}_i^2}{\lambda_i^2} \mathrm{E}_{\pi_d} [\xi_i^2 \tilde{\alpha}^-(x,y) ].
\end{align*}
Since $\tilde{\alpha}^-(x,y) \in (0,1]$, it follows from Jensen's inequality that 
$$
	| \mathrm{E}_{\pi_d}[\xi_i \tilde{\alpha}^-(x,y)] | \leq \mathrm{E}_{\pi_d}[|\xi_i| \tilde{\alpha}^-(x,y)] \leq \mathrm{E}_{\pi_d}[|\xi_i|] \leq \mathrm{E}_{\pi_d}[\xi_i^2]^{1/2} = \gamma_{d,i}^{1/2}.
$$
Also, $0 \leq \mathrm{E}_{\pi_d} [\xi_i^2 \tilde{\alpha}^-(x,y) ] \leq \mathrm{E}_{\pi_d} [\xi_i^2 ] = \gamma_{d,i}$.  

Finally, using $z \mapsto 1\wedge \exp(z)$ Lipschitz, and since $T_{ji}$ are $\mathcal{O}(d^{-1/2})$,
$$
	|\mathrm{E}_{\pi_d}[ \tilde{\alpha}^-(x,y)] - \mathrm{E}_{\pi_d}[ \tilde{\alpha}(x,y) ]|
	\leq \mathrm{E}_{\pi_d}|Z_{d,i}| \rightarrow 0 \qquad \mbox{as $d \rightarrow \infty$}.
$$
Since $\tilde{g}_i^2 \hat{r}_i^2 + \frac{g_i}{\tilde{\lambda}_i^2} = \mathcal{O}(t_{d,i}\lambda_i^{-1})$ it follows that as $d \rightarrow \infty$
\begin{align*}
	S_d &= S_d^- + \mathrm{o}(t_{d,i} \lambda_i^{-1}) \\
	&= \left( \left(\tilde{g}_i^2 \hat{r}_i^2 + \frac{g_i}{\tilde{\lambda}_i^2} \right) \mathrm{E}_{\pi_d}[\tilde{\alpha}(x,y)] + 2 \frac{\hat{r}_i \tilde{g}_i^2 \gamma_{d,i}^{1/2}}{\lambda_i} u_{d,i} + \frac{\tilde{g}_i^2 \gamma_{d,i}}{\lambda_i^2} v_{d,i} \right) + \mathrm{o}(t_{d,i} \lambda_i^{-1}) 
\end{align*}
for some  $u_{d,i} \in [-1,1]$ and $v_{d,i} \in [0,1]$.  The result then follows from \eqref{eq:acceptlim} and $\tilde{g}_i^2 \hat{r}_i^2 + \frac{g_i}{\tilde{\lambda}_i^2} = \mathcal{O}(t_{d,i}\lambda_i^{-1})$.
\end{proof}

\chapter{Examples}
\label{sec examples}

\section{Discretized Langevin diffusion - MALA and SLA}
\label{sec langevin}

The proposal for MALA is obtained from the Euler-Maruyama discretization of a Langevin diffusion process $\{ z_t \}$ which satisfies the stochastic differential equation
\[
	\frac{\dd z_t}{\dd t} = \frac{1}{2} \nabla \log \pi(z_t) + \frac{\dd W_t}{\dd t},
\]
where $W_t$ is standard Brownian motion in $\bbR^d$.  This diffusion process has the desired target distribution $\pi$ as equilibrium, so one might expect a discretization of the diffusion process to almost preserve the desired target distribution.  If the target is Gaussian, $\normal(A^{-1}b,A^{-1})$, then for current state $x \in \bbR^d$ and time step $h >0$, the MALA proposal $y \in \bbR^d$ is defined as
\begin{equation}
\label{MALA prop}
	y = (I - \tsfrac{h}{2} A) x + \tsfrac{h}{2} b + \sqrt{h} \xi 
\end{equation}
where $\xi \sim \normal(0,I)$.  One can also use this proposal for situations where the target distribution is a change of measure from a Gaussian; in which case it is called the SLA proposal \cite{BRS2009}.  The SLA algorithm is
\begin{equation}
\label{SLA}
\begin{split}
	\mbox{Target:} & \qquad \pi_d, \\
	\mbox{Proposal:} & \qquad y = (I - \tsfrac{h}{2} A) x + \tsfrac{h}{2} b + \sqrt{h} \xi, \qquad \mbox{where $\xi \sim \normal(0,I)$}, 
\end{split}
\end{equation}

Identifying \eqref{MALA prop} with \eqref{eq ar1} and applying Corollary \ref{lem equiv} we have the following theorem.  

\begin{theorem}
\label{thm MALA split}
The SLA proposal corresponds to the matrix splitting
\begin{align*}
	M &= \tsfrac{2}{h}(I - \tsfrac{h}{4}A), & 
	\mathcal{A} &= (I - \tsfrac{h}{4}A) A, \\
	N &= \tsfrac{2}{h}(I - \tsfrac{h}{4}A)(I-\tsfrac{h}{2}A), &
	\beta &= (I - \tsfrac{h}{4}A) b.
\end{align*}
\end{theorem}

Thus, the SLA proposal corresponds to a matrix splitting where $M$ and $N$ are functions of $A$ and our theory applies.  An important feature of this proposal is that $\hat{r}_i = 0$ for all $i$.  This greatly simplifies the results in Theorems \ref{thm accept}, \ref{thm jumpsize}, \ref{thm:ng1} and \ref{thm:ng2} and we extend existing theory (\cite[Cor. 1]{BRS2009} and the simpler results in \cite[Thm. 7]{RR2001}) to the case where the reference Gaussian measure is allowed to have off-diagonal covariance terms.  The theorem below is a special case of Theorem \ref{thm genlang conv} so we omit the proof.

\begin{theorem}
\label{thm MALA conv}
Suppose there exist constants $c,C>0$ and $\kappa \geq 0$ such that the eigenvalues $\lambda_i^2$ of $A$ satisfy
$$
	c i^\kappa \leq \lambda_i \leq C i^\kappa \qquad \mbox{for $i=1,\dotsc,d$}.
$$
Also suppose that $\phi_d$ satisfies Assumption \ref{assumption1} and $r = \max \{ s,s',s'' \}$ satisfies Assumption \ref{assumption2}.

If $h = l^2 d^{-1/3-2\kappa}$ for some $l >0$ then SLA, in equilibrium, satisfies
$$
	\mathrm{E}[\alpha(x,y)] \rightarrow 2 \Phi\left( -\tsfrac{l^3 \sqrt{\tau}}{8} \right)
$$
and
\begin{equation}
\label{eq jump}
	\mathrm{E}[(x_i-x_i)^2] = 2 h \Phi\left( -\tsfrac{l^3 \sqrt{\tau}}{8} \right) + \mathrm{o}(h)
\end{equation}
as $d \rightarrow \infty$ where $\tau = \lim_{d \rightarrow \infty} \frac{1}{d^{1 + 6 \kappa}} \sum_{i=1}^d \lambda_i^6$.
\end{theorem}  

Thus, the performance of SLA depends on the choice of $h$ which is usually tuned (by tuning $l$) to maximise the expected jump distance.  From \eqref{eq jump}, using $s = l \tau^{1/6}/2$, we have 
\begin{equation}
\label{eq:mj}
	\max_{l > 0} 2 l^2 d^{-1/3-2\kappa} \Phi\left( -\tsfrac{l^3 \sqrt{\tau}}{8} \right) = \max_{s > 0} \frac{8 d^{-1/3-2\kappa}}{\tau^{1/3}} s^2 \Phi(-s^3),
\end{equation}
which is maximised at $s_0=0.8252$, independent of $\tau$.  Therefore, the acceptance rate that maximises expected jump distance is $2 \Phi(-s_0^3) = 0.574$.  This result was first stated in \cite{RR1998} for product target distributions, then more generally in \cite{RR2001,BRS2009}.  Our result is even more general because we allow the reference measure, which must be Gaussian in our case, to have off-diagonal covariance terms.  In practice, $h$ (equivalently $l$) is adjusted so that the acceptance rate is approximately $0.574$ to maximise the expected jump size.  This acceptance rate is independent of $\tau$, so it is independent of the eigenvalues of $A$.  However, the efficiency of SLA still depends on $\tau$, and as $\tau$ increases the expected square jump distance will decrease by a factor $\tau^{1/3}$.  

The rationale for studying the case when $d \rightarrow \infty$ is that it is a good approximation for cases when $d$ is `large' and finite (see eg. \cite[Fig. 1]{RR1998} or \cite{RR2001}).  However, the results above suggest that we should take $h \rightarrow 0$ as $d \rightarrow \infty$, to achieve at best an expected jump size that also tends towards $0$ as $d \rightarrow \infty$.  The only good point about these results is that SLA has superior asymptotic performance over RWM (see \cite[Thms. 1-4 and Cor. 1]{BRS2009} and \cite[Fig. 1]{RR1998}).

To understand the convergence of SLA to equilibrium (burn in) we would like to know the `spectral gap' or second largest eigenvalue of the transition kernel, as this determines the rate of convergence.  As far as we are aware this is an open problem.

We can also use Theorem \ref{thm MALA split} to analyse the unadjusted Langevin algorithm (ULA) in \cite{RT1996} in the case when the target is Gaussian, in which case ULA is simply the proposal chain from \eqref{MALA prop} without the accept/reject step in the MH algorithm.  From Theorem \ref{thm MALA split} we see that it does not converge to the correct target distribution since $\mathcal{A} \neq A$.  Instead of converging to $\normal(A^{-1}b, A^{-1})$, ULA converges to $\normal(A^{-1}b,\mathcal{A}^{-1})$.  Indeed, the authors of \cite{RT1996} note that ULA has poor convergence properties.  We see here the reason why it converges to the wrong target distribution, and from \cite{F2013,FP2016} we know its convergence rate to the incorrect target distribution depends on the spectral radius of $G = I-\tsfrac{h}{2}A$, which is close to $1$ when $h$ is small.

Despite possibly having slow convergence per iteration, the SLA proposal is cheap to compute.  Since $G = I-\tsfrac{h}{2}A$, we only require a single matrix-vector multiplication with $A$ for each proposal and since $\Sigma = h I$, an i.i.d. sample from $\normal(0,\Sigma)$ at each iteration is also cheap to compute.

\section{Discretized Langevin diffusion - more general algorithms}

It is possible to generalise MALA and SLA by `preconditioning' the Langevin diffusion process and using a discretization scheme that is not Euler-Maruyama.  For symmetric positive definite matrix $V \in \bbR^{d \times d}$ (the `preconditioner') consider a Langevin process $\{z_t\}$ satisfying the stochastic differential equation
\begin{equation}
\label{gen langevin}
	\frac{\dd z_t}{\dd t} = \frac{1}{2} V \nabla \log \pi(z_t) + \frac{\dd W_t}{\dd t}
\end{equation}
where $W_t$ is Brownian motion in $\bbR^d$ with covariance $V$.  This diffusion process also preserves $\pi$.  For $\theta \in [0,1]$, time step $h >0$ and current state $x \in \bbR^d$ define a proposal $y \in \bbR^d$ by discretizing \eqref{gen langevin} as
$$
	y - x = \tsfrac{h}{2} V \nabla \log \pi(\theta y + (1-\theta)x) + \sqrt{h} \nu
$$
where $\nu \sim \normal(0,V)$.  When the target is Gaussian, $\normal(A^{-1}b,A^{-1})$, this can be rewritten as 
\begin{equation}
\label{gen langevin prop}
	y = (I + \tsfrac{\theta h}{2} VA)^{-1} \left[ (I - \tsfrac{(1-\theta)h}{2} VA)x + \tsfrac{h}{2} V b + (h V)^{1/2} \xi \right]
\end{equation}
where $\xi \sim \normal(0,I)$, which is an AR(1) proposal.  Thus, we define a MH algorithm by
\begin{equation}
\label{GenLang}
\begin{split}
	\mbox{Target:} & \; \pi_d, \\
	\mbox{Proposal:} & \; y = (I + \tsfrac{\theta h}{2} VA)^{-1} \left[ (I - \tsfrac{(1-\theta)h}{2} VA)x + \tsfrac{h}{2} V b + (h V)^{1/2} \xi \right] \!\! , \; \mbox{for $\xi \sim \normal(0,I)$}.
\end{split}
\end{equation}

Different choices of $\theta$ and $V$ give different proposals.  For example, SLA has $\theta = 0$ and $V = I$, and pCN corresponds to $\theta = \tsfrac{1}{2}$ and $V = A^{-1}$.  Table \ref{tab1} describes several more examples.

\begin{table}
\begin{center}
\begin{tabular}{|c|c|p{12cm}|}
\hline
$\theta$ & $V$ & Method \\
\hline
$0$ & $I$ & SLA, MALA and ULA \cite{RT1996, BRS2009}.  SLA=MALA for Gaussian target distributions. \\ \hline
$0$ & $A^{-1}$ & Used in \cite{PST2012}. With a change of variables $x \leftrightarrow Q^T V^{-1/2}x$, it is the Preconditioned Simplified Langevin Algorithm (P-SLA) \cite{BRS2009}. \\ \hline
$\in [0,1]$ & $I$ & $\theta$-SLA \cite{BRS2009}.  \\ \hline
$\tsfrac{1}{2}$ & $I$ & CN \cite{CRSW2013}. \\ \hline
$\tsfrac{1}{2}$ & $A^{-1}$ & pCN \cite{CRSW2013}.  
\\ \hline
\end{tabular}
\end{center}
\caption{Different choices of $\theta$ and $V$ in \eqref{gen langevin prop} lead to different proposals for the MH algorithm \eqref{GenLang}.  Other choices are possible.}
\label{tab1}
\end{table}

Applying Corollary \ref{lem equiv} we can prove the following theorem.

\begin{theorem}
\label{thm genlang}
The general Langevin proposal \eqref{gen langevin prop} corresponds to the matrix splitting
\begin{align*}
	M &= \tsfrac{2}{h} V^{-1/2} W (I + \tsfrac{\theta h}{2} B ) V^{-1/2}, &
	\mathcal{A} &= V^{-1/2} W B V^{-1/2} = \tilde{W} A, \\
	N &= \tsfrac{2}{h} V^{-1/2} W (I - \tsfrac{(1-\theta)h}{2} B) V^{-1/2}, &
	\beta &= V^{-1/2} W V^{1/2} b = \tilde{W} b,
\end{align*}
where $B = V^{1/2} A V^{1/2}$, $W = I + (\theta-\tsfrac{1}{2}) \tsfrac{h}{2} B$ and $\tilde{W} = I + (\theta-\tsfrac{1}{2}) \tsfrac{h}{2} AV$.
\end{theorem}

\begin{proof}
With iteration matrix $G = (I + \frac{\theta h}{2}VA)^{-1}(I - \tsfrac{(1-\theta)h}{2} VA)$, vector $g = \tsfrac{h}{2}(I + \frac{\theta h}{2}VA)^{-1} V b$ and matrix $\Sigma = (I + \frac{\theta h}{2}VA)^{-1} (hV) (I + \frac{\theta h}{2}AV)^{-1}$, then \eqref{gen langevin prop} is the same as \eqref{eq ar1}.  To apply Corollary \ref{lem equiv} we first check that $G \Sigma$ is symmetric.  We have
$$
	G = V^{1/2} (I + \tsfrac{\theta h}{2} B)^{-1} (I - \tsfrac{(1-\theta)h}{2}B) V^{-1/2}
	\quad \mbox{and} \quad
	\Sigma = h V^{1/2} (I + \tsfrac{\theta h}{2}B)^{-2} V^{1/2}
$$
so that 
\begin{align*}
	G \Sigma = h V^{1/2} (I + \tsfrac{\theta h}{2}B)^{-1} (I - \tsfrac{(1-\theta)h}{2}B) (I + \tsfrac{\theta h}{2}B)^{-2} V^{1/2}
\end{align*}
which is symmetric since $V$ and $B$ are symmetric.  Applying Corollary \ref{lem equiv} then yields the result.
\end{proof}

Because of Theorem \ref{thm genlang}, the proposal chain for \eqref{gen langevin prop} converges to $\normal(A^{-1}b, (\tilde{W}A)^{-1})$, and if $\theta \neq \tsfrac{1}{2}$, then $\tilde{W} \neq I$, $\mathcal{A} \neq A$, and the target reference and proposal limit distributions disagree.  

If $\theta = \tsfrac{1}{2}$, then the proposal limit and target reference distributions are the same.  If, in addition, the target is Gaussian, then the MH accept/reject step is redundant and we can use \cite{F2013,FP2016,FP2014} to analyse and accelerate the proposal chain generated by \eqref{gen langevin prop}.  

To evaluate the performance of the MH algorithm \eqref{GenLang} when $\mathcal{A} \neq A$ we would like to be able to apply Theorems \ref{thm:ng1} and \ref{thm:ng2}, but we see in Theorem \ref{thm genlang} that the splitting matrices are functions of $B$ and $V$ (not $A$) so we cannot directly apply our new theory.  However, a change of coordinates will fix this!  The following lemma is a result of simple algebra and Theorem \ref{thm genlang}.

\begin{lemma}
\label{lem ch1}
Under the change of coordinates 
$$
	x \leftrightarrow V^{-1/2} x
$$
the MH algorithm \eqref{GenLang} is transformed to the MH algorithm defined by
\begin{equation}
\label{eq hat1}
\begin{split}
	\mbox{Target:} & \quad \pi_d \quad \mbox{where } \tsfrac{\dd \pi_d}{\dd \tilde{\pi}_d}(x) = \exp( - \phi_d(V^{1/2}x)) \mbox{ and }\tilde{\pi}_d \mbox{ is }  \normal(B^{-1} V^{1/2} b, B^{-1}), \\
	\mbox{Proposal:} & \quad (I + \tsfrac{\theta h}{2} B) y = (I - \tsfrac{(1-\theta)h}{2} B) x + \tsfrac{h}{2} V^{1/2} b + h^{1/2} \xi, \quad \mbox{where $\xi \sim \normal(0,I)$}, 
\end{split}
\end{equation}
and $B = V^{1/2} A V^{1/2}$.  Moreover, the proposal in \eqref{eq hat1} corresponds to the matrix splitting $B = M-N$ where
\begin{align*}
	M &= \tsfrac{2}{h} W (I + \tsfrac{\theta h}{2} B ), &
	\mathcal{A} &= W B, \\
	N &= \tsfrac{2}{h} W (I - \tsfrac{(1-\theta)h}{2} B), &
	\beta &= W V^{1/2} b,
\end{align*}
and $W = I + (\theta-\tsfrac{1}{2}) \tsfrac{h}{2} B$.
\end{lemma}

Thus, we have transformed the MH algorithm \eqref{GenLang} to a MH algorithm where the splitting matrices are functions of the target reference precision matrix, and we can apply Theorems \ref{thm:ng1} and \ref{thm:ng2} to \eqref{eq hat1} to find the expected acceptance rate and expected jump size of \eqref{GenLang}.  Note that we never compute the Markov chain for \eqref{eq hat1}, we only use it to determine the convergence properties of \eqref{GenLang} since they are identical.

\begin{theorem}
\label{thm genlang conv}
Suppose there are constants $c,C > 0$ and $\kappa \geq 0$ such that the eigenvalues $\lambda_i^2$ of $B = V^{1/2} A V^{1/2}$ (equivalently, $\lambda_i^2$ are eigenvalues of $VA$) satisfy
\[
	c i^\kappa \leq \lambda_i \leq C i^\kappa \qquad \mbox{for $i=1,\dotsc,d$}.
\]
Also suppose that $\psi_d(x) := \phi_d(V^{1/2}x)$ satisfies Assumption \ref{assumption1} (with $|\cdot|_s = |B^s \cdot|$ for $s \in \bbR$) and $r = \max \{ s,s',s'' \}$ satisfies Assumption \ref{assumption2}.

If $h = l^2 d^{-1/3 - 2 \kappa}$ for $l > 0$ and $\tau = \lim_{d\rightarrow \infty} \frac{1}{d^{6\kappa + 1}} \sum_{i=1}^d \lambda_i^6$ then MH algorithm \eqref{GenLang}, in equilibrium, satisfies
\begin{equation}
\label{eq expect}
	\mathrm{E}[\alpha(x,y)] \rightarrow 2 \Phi\left( -\frac{l^3 |\theta-\tsfrac{1}{2}| \sqrt{\tau}}{4} \right) 
\end{equation}
and for normalised eigenvector $q_i$ of $B$ corresponding to $\lambda_i^2$,
\begin{equation}
\label{eq jump2}
	\mathrm{E}[ |q_i^T V^{-1/2} (x' - x)|^2 ] = 2 h \Phi\left( -\frac{l^3 |\theta-\tsfrac{1}{2}| \sqrt{\tau}}{4} \right) + \mathrm{o}(h)
\end{equation}
as $d \rightarrow \infty$.
\end{theorem}

The proof of Theorem \ref{thm genlang conv} is in the Appendix.

For efficiency, as well as considering the expected squared jump size, we must also consider the computing cost of the proposal \eqref{gen langevin prop}, which requires the action of $(I + \frac{\theta h}{2}VA)^{-1}$ and an independent sample from $\normal(0,V)$, as well as the actions of $V$ and $A$ multiplying a vector.

We see that $V$ plays a similar role to a preconditioner in solving a linear system.  For a linear system we choose $V$ to be cheap to compute matrix multiplication and to minimise the condition number of $VA$.  For the MH algorithm \eqref{GenLang} we choose it so that  multiplying with $V$ and sampling from $\normal(0,V)$ are cheap to compute and to minimise $\tau$.  In both cases $V$ is chosen to `control' the eigenvalues of $VA$.

Although SLA and more general discretizations of Langevin diffusion have been successfully analyzed in \cite{BRS2009}, all of these results are stated for target distributions that are a change of measure from product measures.  Theorem \ref{thm genlang conv} extends their theory (in particular \cite[Cor. 1]{BRS2009}) to the case where the reference measure $\tilde{\pi}_d$ may be Gaussian with off-diagonal covariance terms and $\theta \in [0,1]$.  See also \cite[Thm. 7]{RR2001}.

With $\theta = 0$ and $V = I$, SLA is very cheap to compute because we only have to invert the identity matrix and sample from $\normal(0,I)$ at each iteration (as well as multiply by $A$).  Alternatively, pCN, with $\theta = \tsfrac{1}{2}$ and $V = A^{-1}$, requires a sample from $\normal(0,A^{-1})$ which may be computationally expensive, particularly in high dimensions.

\section{$L$-step methods}

Given an AR(1) proposal of the form \eqref{eq ar1}, we can form a new AR(1) proposal by taking $L$ steps of the original proposal before performing the MH accept/reject step.  This may be advantageous when the cost of evaluating $\phi_d$ is significant.  The $L$-step proposal is computed by iterating
$$
	y^{(l)} = G y^{(l-1)} + g + \nu^{(l)} \qquad \mbox{for $l=1,\dotsc,L$},
$$
where $\nu^{(l)}$ is an i.i.d. draw from $\normal(0,\Sigma)$ and $y^{(0)} = x$.  This yields a new proposal in the form of \eqref{eq ar1},
\begin{equation}
\label{eq:lprop}
	y = G_L x + g_L + \nu_L, \qquad \mbox{with } \nu_L \sim \normal(0,\Sigma_L)
\end{equation}
where $G_L = G^L$, $g_L = (I-G)^{-1}(I-G^L)g$ and $\Sigma_L = \sum_{l=0}^{L-1} G^l \Sigma (G^T)^l$.  Hence, the eigenvalues of $G_L$ are $G_i^L$, and if $G_i < 1$ then the $L$-step proposal chain will converge to the same limit as the $1$-step proposal chain (i.e. $\mathcal{A}_L = \mathcal{A}$ and $\beta_L = \beta$).  

We can reduce the computational cost of evaluating the acceptance ratio for the $L$-step proposal using the surrogate transition method \cite[p.194]{Liu2001book}.  The proof of the following lemma is in the Appendix.

\begin{lemma}
\label{lem:lstepalpha}
The $L$-step acceptance probability satisfies
$$
	\alpha(x,y) 
	= 1 \wedge \frac{ \pi_d(y) q_L(y,x)}{ \pi_d(x) q_L(x,y) }
	= 1 \wedge \frac{ \pi_d(y) \pi^*(x)}{\pi_d(x) \pi^*(y)}
$$
where $q_L(x,dy) = q_L(x,y) \dd y$ is the transition kernel for the $L$-step proposal $y$ given $x$ from \eqref{eq:lprop} and $\pi^*(x) \propto \exp( -\frac{1}{2} x^T \mathcal{A} x + \beta^T x )$.  
\end{lemma}


The computational cost of the $L$-step proposal is $L$ times the cost of the original proposal, but the expected squared jump size for the $L$-step method is, in general, not $L$ times the original.  For example, let us consider $L$-step SLA, where $G = (I-\tsfrac{h}{2}A)$, $\beta = \tsfrac{h}{2}b$ and $\Sigma = hI$, and the proposal is given by \eqref{eq:lprop}.  The proof of the following theorem is in the Appendix.

\begin{theorem}
\label{thm:lstep}
Suppose there exist constants $c,C>0$ and $\kappa \geq 0$ such that the eigenvalues $\lambda_i^2$ of $A$ satisfy
$$
	c i^\kappa \leq \lambda_i \leq C i^\kappa \qquad \mbox{for $i=1,\dotsc,d$}.
$$
Also suppose that $\phi_d$ satisfies Assumption \ref{assumption1} and $r = \max \{ s,s',s'' \}$ satisfies Assumption \ref{assumption2}.

If $h = l^2 d^{-1/3-2\kappa}$ for some $l >0$ then $L$-step SLA, in equilibrium, satisfies
\begin{equation}
\label{eq:lstepa}
	\mathrm{E}[\alpha(x,y)] \rightarrow 2 \Phi\left( - \tsfrac{l^3 \sqrt{L\tau}}{8} \right)
\end{equation}
and
\begin{equation}
\label{eq ljump}
	\mathrm{E}[(x_i'-x_i)^2] = 2 L h \Phi\left( -\tsfrac{l^3 \sqrt{L\tau}}{8} \right) + \mathrm{o}(h)
\end{equation}
as $d \rightarrow \infty$ where $\tau = \lim_{d \rightarrow \infty} \frac{1}{d^{1 + 6 \kappa}} \sum_{i=1}^d \lambda_i^6$.
\end{theorem}

To maximise the performance of $L$-step SLA we then tune $l$ to maxmise the expected jump size.  From \eqref{eq ljump}, using $s = l (L\tau)^{1/6}/2$, we have
\begin{equation}
\label{eq:lmj}
	\max_{l > 0} 2 L l^2 d^{-1/3 - 2 \kappa} \Phi\left( -\tsfrac{l^3 \sqrt{L\tau}}{8} \right) = \max_{s>0} 8 L^{2/3} d^{-1/3 - 2 \kappa} \tau^{-1/3} s^2 \Phi( - s^3 ),
\end{equation}
which is maxmised at $s_0 = 0.8252$.  Therefore, the expected jump size of $L$-step SLA is maxmimised when the acceptance rate is $2 \Phi(-s_0^3) = 0.574$, which is the same as SLA, but this corresponds to an expected jump size that is only $L^{2/3}$ times larger than the jump size for SLA (compare \eqref{eq:lmj} and \eqref{eq:mj}) in the limit when $d \rightarrow \infty$.

To compare the efficiency $L$-step SLA for varying $L$ we must also consider the computational cost of the method.  For example, suppose that matrix-vector products with $A$ cost $1$ unit of CPU time, inner products and drawing independent samples from $\normal(0,I)$ are essentially free, and evaluating $\phi_d$ costs $t$ units of CPU time.  From Lemma \ref{lem:lstepalpha}, we can simplify the acceptance ratio for $L$-step SLA to
$$
	\alpha(x,y) = 1 \wedge \exp\left( \tsfrac{h}{4} (|A x|^2 - |Ay|^2) - \tsfrac{h}{2} b^T (Ax-Ay) + \phi_d(x) - \phi_d(y) \right) 
$$
so $L$-step SLA uses $L$ matrix vector products with $A$ per proposal and an additional matrix-vector product and two evaluations of $\phi_d$ in the acceptance ratio.  If the proposal is accepted then we can reuse some of the calculations in the acceptance ratio, but if it is rejected then a matrix-vector product and an evaluation of $\phi_d$ are wasted.  The average cost of an $L$-step SLA iteration is then
$$
	L + t + (1-\alpha)(1 + t) = 1.426 + 0.426 t + L
$$
units of CPU time, assuming that we have tuned $L$-step SLA so that the acceptance rate is $0.574$.  Also let $1$ unit of jump size be the expected jump size of $1$-step SLA, then $L$-step SLA has an expected jump size of $L^{2/3}$ units, and the `efficiency' of $L$-step SLA is calculated as jump size divided by computing cost,
$$
	\frac{L^{2/3}}{1.426 + 0.426 t + L},
$$
which is maxmised at $L = 2(1.426 + 0.426 t)$.  Our conclusion is that SLA can be improved by using $L$-step SLA with $L > 1$, and the optimal value of $L$ depends on the cost of evaluating $\phi_d$.  If $t = 0$, then $L=3$ is optimal.  Figure \ref{fig:1} shows the efficiency of $L$-step SLA for other values of $t$.  

This analysis can be repeated for other $L$-step algorithms.

\begin{figure}
\begin{center}
%
%
%
%
\begin{tikzpicture}

\begin{axis}[%
view={0}{90},
width = 3.6in,
height = 2.7in,
scale only axis,
xmin=0, xmax=11,
xlabel={$L$},
ymin=0.2, ymax=0.5,
ylabel={efficiency},
axis lines*=left,
legend style={at={(0.97,0.03)},anchor=south east,align=left}]
\addplot [
color=black,
only marks,
mark=o,
mark options={solid}
]
coordinates{
 (1,0.412201154163232)(2,0.463339478099299)(3,0.46996923250156)(4,0.464401419054505)(5,0.455029215408165)(6,0.444644121854919)(7,0.434287409212316)(8,0.424358158285593)(9,0.414996039796876)(10,0.406230424786695) 
};
\addlegendentry{$t=0$};

\addplot [
color=black,
only marks,
mark=square,
mark options={solid}
]
coordinates{
 (1,0.350631136044881)(2,0.412097884726947)(3,0.428706476309131)(4,0.430595027305151)(5,0.426739307970354)(6,0.420520536028353)(7,0.413387450296314)(8,0.406008932196508)(9,0.398705188990253)(10,0.391629162471547) 
};
\addlegendentry{$t=1$};

\addplot [
color=black,
only marks,
mark=triangle,
mark options={solid}
]
coordinates{
 (1,0.305064063453325)(2,0.371061489473632)(3,0.394104551544506)(4,0.401376568937519)(5,0.401761162161702)(6,0.398879831951513)(7,0.39440673744589)(8,0.389180774469741)(9,0.383645035549053)(10,0.378041116925621) 
};
\addlegendentry{$t=2$};

\addplot [
color=black,
only marks,
mark=diamond,
mark options={solid}
]
coordinates{
 (1,0.219490781387182)(2,0.285709332607667)(3,0.317279411691871)(4,0.333488896213571)(5,0.34175055378832)(6,0.345534454677127)(7,0.346656471203389)(8,0.346140533056421)(9,0.344596106317476)(10,0.342401064739804) 
};
\addlegendentry{$t=5$};

\addplot [
color=blue,
only marks,
mark=*,
mark options={solid,fill=black,draw=black},
forget plot
]
coordinates{
 (3,0.46996923250156) 
};
\addplot [
color=blue,
only marks,
mark=square*,
mark options={solid,fill=black,draw=black},
forget plot
]
coordinates{
 (4,0.430595027305151) 
};
\addplot [
color=blue,
only marks,
mark=triangle*,
mark options={solid,fill=black,draw=black},
forget plot
]
coordinates{
 (5,0.401761162161702) 
};
\addplot [
color=blue,
only marks,
mark=diamond*,
mark options={solid,fill=black,draw=black},
forget plot
]
coordinates{
 (7,0.346656471203389) 
};
\end{axis}
\end{tikzpicture}%
\end{center}
\caption{Efficiency of the $L$-step SLA method for varying number of steps $L$ and varying computing cost for evaluating $\phi_d$.  Filled markers correspond to maximum efficiency.}
\label{fig:1}
\end{figure}
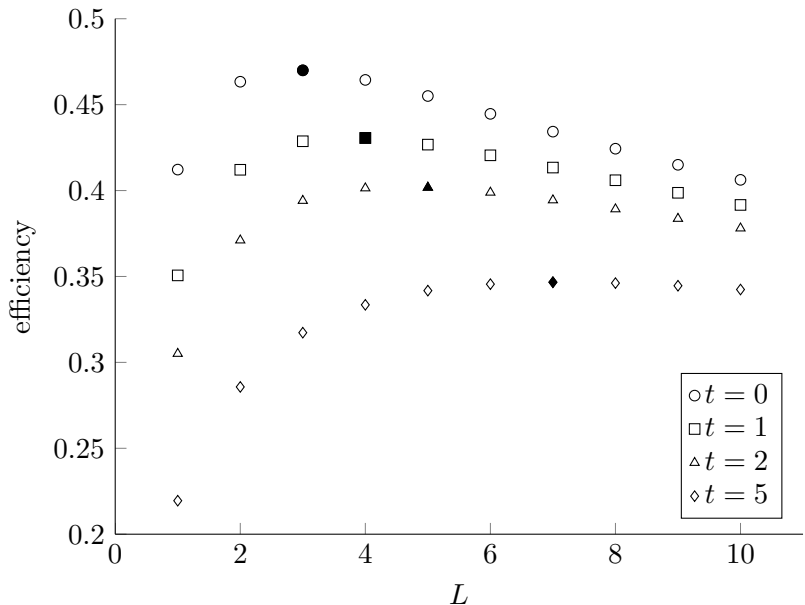

\section{Hybrid Monte Carlo}

Another type of AR(1) proposal, that fits our theory when $\pi_d$ is Gaussian, are proposals from the Hybrid (or Hamiltonian) Monte Carlo algorithm (HMC), see e.g. \cite{DKPR1987,BPRSS2010,N1993}.  For this section, suppose that the target $\pi_d$ is the Gaussian $\normal(A^{-1}b,A^{-1})$.

HMC treats the current state $x \in \bbR^d$ as the initial position of a particle, the initial momentum $p \in \bbR^d$ of the particle is chosen independently at random, and then the motion of the particle is evolved according to a Hamiltonian system for a fixed amount of time.  The final position of the particle is the proposal.  Instead of solving the Hamiltonian system exactly, the evolution of the particle is approximated using a reversible, symplectic numerical integrator.  For example, the leap-frog method (also called the Stormer-Verlet method) is an integrator that preserves a modified Hamiltonian, see e.g. \cite{HLW}.  Hence, the proposal $y$, for target $\normal(A^{-1}b,A^{-1})$, is computed as follows; let $V \in \bbR^{d \times d}$ be a symmetric positive definite matrix and define a Hamiltonian function $H : \bbR^d \times \bbR^d \rightarrow \bbR$ by
\begin{equation}
\label{eq ham1}
	H(q,p) := \frac{1}{2} p^T V p + \frac{1}{2} q^T A q - b^T q. 
\end{equation}
Given a time step $h>0$, a number of steps $L \in \mathbb{N}$, and current state $x \in \bbR^d$, define $q_0 := x$, and sample $p_0 \sim \normal(0,V^{-1})$.  Then for $l=0,\dotsc,L-1$ compute
\begin{align*}
	p_{l+1/2} &= p_{l} - \tsfrac{h}{2} (A q_{l} - b), \\
	q_{l+1} &= q_{l} + h V p_{l+1/2}, \\
	p_{l+1} &= p_{l+1/2} - \tsfrac{h}{2} (A q_{l+1} - b).
\end{align*}
The proposal is then defined as $y := q_L$.  In matrix form we have
$$
	\left[ \twobyone{q_{l+1}}{p_{l+1}} \right] = K 	\left[ \twobyone{q_l}{p_l} \right] + J \left[ \twobyone{0}{\tsfrac{h}{2} b} \right]
$$
where $K,J \in \bbR^{2d \times 2d}$ are defined as
$$
	K = 
	\left[ \!\! \twobytwo{I}{0}{-\tsfrac{h}{2} A}{I} \right]
	\left[ \twobytwo{I}{h V}{0}{I}  \right]
	\left[ \!\! \twobytwo{I}{0}{-\tsfrac{h}{2} A}{I} \right]
	=\left[ \!\! \twobytwo{I-\tsfrac{h^2}{2}VA}{hV}{-h A + \tsfrac{h^3}{4}AVA}{I-\tsfrac{h^2}{2}AV} \!\right]
$$
and
$$
	J = 	
	\left[ \twobytwo{I}{0}{0}{I} \right]	
	+	
	\left[ \twobytwo{I}{0}{-\tsfrac{h}{2} A}{I} \right]
	\left[ \twobytwo{I}{h V}{0}{I} \right]
	=\left[ \twobytwo{2I}{hV}{-\tsfrac{h}{2}A}{2I - \tsfrac{h^2}{2}AV} \right].
$$
Hence, $y$ is given by 
\begin{equation}
\label{hmc prop2}
	\left[ \twobyone{y}{p_L} \right]
	=
	K^L \left[ \twobyone{x}{\xi} \right] + \sum_{l=0}^{L-1} K^l J \left[ \twobyone{0}{\tsfrac{h}{2}b} \right] 
	\quad\mbox{where $\xi \sim \normal(0,V)$},
\end{equation}
or equivalently, 
\begin{equation}
\label{hmc prop}
	y = (K^L)_{11} x + \left( SJ \left[ \twobyone{0}{\tsfrac{h}{2}b} \right] \right)_1 + (K^L)_{12} \xi, \qquad \mbox{where $\xi \sim \normal(0,V^{-1})$,}
\end{equation}
and $(K^L)_{ij}$ is the $ij$ block (of size $d\times d$) of $K^L$, $S = (I-K)^{-1} (I-K^L)$ and $(\cdot)_1$ are the first $d$ entries of the vector $(\cdot)$.

In the case of only one time step of the leap-frog integrator ($L=1$) then HMC is MALA \cite{BPRSS2010}.  Hence, we immediately know that the HMC proposal with $L=1$ is an AR(1) proposal where the proposal limit and target distributions are not the same, and the expected acceptance rate and jump size are given by Theorem \ref{thm MALA conv}.  The case for $L>1$ is more complicated, but \eqref{hmc prop} is still an AR(1) proposal that can be expressed as a matrix splitting using \eqref{eq ar1b}.  The proofs of the following two results are in the Appendix.

\begin{theorem}
\label{thm HMC1}
The HMC proposal \eqref{hmc prop} corresponds to the matrix splitting 
\begin{align*}
	M &= \Sigma^{-1} (I+(K^L)_{11}), &
	\mathcal{A} &= \Sigma^{-1} (I - (K^L)_{11}^2), \\
	N &= \Sigma^{-1} (I+(K^L)_{11}) (K^L)_{11}, &
	\beta &= \Sigma^{-1} (I+(K^L)_{11}) \left(( SJ \left[ \twobyone{0}{\tsfrac{h}{2}b} \right] \right)_1, 
\end{align*}	
where $\Sigma = (K^L)_{12} V^{-1} (K^L)_{12}^T$.
\end{theorem}

\begin{corollary}
\label{cor 15}
The matrix splitting from HMC satisfies $\mathcal{A}^{-1} \beta \! = A^{-1} b$.
\end{corollary}

These results imply that the proposal chain for HMC converges to $\normal(A^{-1}b, \mathcal{A}^{-1})$ where $\mathcal{A} \neq A$, rather than the desired target $\normal(A^{-1}b,A^{-1})$, so the MH accept/reject step is necessary even when the target is Gaussian.

For the analysis of HMC we require the eigenvalues of the iteration matrix.  A proof of the following result is in the Appendix.

\begin{theorem}
\label{thm uhmc conv}
Let $\lambda_i^2$ be eigenvalues of $B = V^{1/2} A V^{1/2}$ or $VA$ (these matrices have the same eigenvalues).  Then the iteration matrix $G = (K^L)_{11}$ for the HMC proposal has eigenvalues 
$$
	G_i = \cos( L \theta_i )
$$
where $\theta_i = -\cos^{-1} ( 1 - \tsfrac{h^2}{2} \lambda_i^2 )$.
\end{theorem}

From this theorem we see how the eigenvalues of the iteration matrix depend on $V$, the number of time steps $L$, and the time step $h$.  Again we refer to $V$ as a preconditioner (as in \cite{BPSSS2011}) because it plays a similar role to a preconditioner for solving linear systems of equations.  Alternatively, $V$ may be referred to as a mass matrix since $p$ in the Hamiltonian \eqref{eq ham1} is momentum and $H$ is energy.

To complete our analysis of HMC we restrict our attention to the case when $d \rightarrow \infty$ and try to apply Theorems \ref{thm accept} and \ref{thm jumpsize}.  These theorems require that the splitting matrices are functions of the target precision matrix.  A simple change of coordinates achieves this.

\begin{theorem}
\label{thm iso hmc}
Under the change of coordinates 
$$
	\left[ \twobyone{x}{p} \right] \leftrightarrow \mathcal{V}^{-1} \left[ \twobyone{x}{p} \right], \qquad \mbox{where } 
	\mathcal{V} = \left[ \twobytwo{V^{1/2}}{0}{0}{V^{-1/2}} \right] \in \bbR^{2d \times 2d},
$$
the Hamiltonian \eqref{eq ham1} and HMC with target $\normal(A^{-1}b,A^{-1})$ and proposal \eqref{hmc prop} are transformed to a Hamiltonian, and MH algorithm defined by
\begin{equation}
\label{eq hmchat1}
\begin{split}
	\mbox{Hamiltonian:} & \quad \mathcal{H}(x,p):= \tsfrac{1}{2} p^T p + \tsfrac{1}{2} x^T B x - (V^{1/2}b)^T x, \\
	\mbox{Target:} & \quad \normal(B^{-1} V^{1/2} b, B^{-1}), \\
	\mbox{Proposal:} & \quad 
	y = (\mathcal{K}^L)_{11} x + \left( \mathcal{S} \mathcal{J} \left[ \twobyone{0}{\tsfrac{h}{2} V^{1/2} b} \right] \right)_1 + (\mathcal{K}^L)_{12} \xi, \; \mbox{for $\xi \sim \normal(0,I)$},
\end{split}
\end{equation}
where $B = V^{1/2} A V^{1/2}$, $\mathcal{S} = (I-\mathcal{K})^{-1}(I-\mathcal{K}^L)$,
$$
	\mathcal{K} = \left[ \twobytwo{I-\tsfrac{h^2}{2}B}{hI}{-h B + \tsfrac{h^3}{4}B^2}{I-\tsfrac{h^2}{2}B} \right]
	\qquad \mbox{and} \qquad
	\mathcal{J} = \left[ \twobytwo{2I}{hI}{-\tsfrac{h}{2}B}{2I - \tsfrac{h^2}{2}B} \right].
$$
Moreover, the proposal in \eqref{eq hmchat1} corresponds to the matrix splitting $\mathcal{A} = M-N$ with
\begin{align*}
	M \!&=\! (\mathcal{K}^L)_{12}^{-2} (I\!+\!(\mathcal{K}^L)_{11}),\!\! &
	\mathcal{A} \!&=\! (\mathcal{K}^L)_{12}^{-2} (I \!-\! (\mathcal{K}^L)_{11}^2), \\
	N \!&=\! (\mathcal{K}^L)_{12}^{-2} (I\!+\!(\mathcal{K}^L)_{11}) (\mathcal{K}^L)_{11},\!\! &
	\beta \! &=\! (\mathcal{K}^L)_{12}^{-2} (I\!+\!(\mathcal{K}^L)_{11}) \left( \!\mathcal{S} \mathcal{J} \!\left[\! \twobyone{0}{\tsfrac{h}{2}V^{1/2} b} \!\right] \right)_1 \!. 
\end{align*}	
\end{theorem}

\begin{proof}
Use $K = \mathcal{V} \mathcal{K} \mathcal{V}^{-1}$ and $J = \mathcal{V} \mathcal{J} \mathcal{V}^{-1}$.
\end{proof}

Similar coordinate transformations are used in classical mechanics \cite[p. 103]{A1989}, see also \cite{BCSS2014}.

MH algorithm \eqref{eq hmchat1} has splitting matrices that are functions of the target precision matrix, so we can apply Theorems \ref{thm accept} and \ref{thm jumpsize} to \eqref{eq hmchat1} to reveal information about the performance of the original HMC algorithm.  A proof of the following result is in the Appendix.

\begin{theorem}
\label{thm hmc conv}
Suppose there are constants $c,C>0$ and $\kappa \geq 0$ such that the eigenvalues of $B = V^{1/2}AV^{1/2}$ (equivalently, $V A$ or $AV$) satisfy
$$
	c i^\kappa \leq \lambda_i \leq C i^\kappa \qquad \mbox{for $i=1,\dotsc,d$}.
$$
If $h = l d^{-1/4-\kappa}$ for $l > 0$, and $L = \lfloor \tsfrac{T}{h} \rfloor$ for fixed $T$, then the HMC algorithm (with proposal \eqref{hmc prop} and target $\normal(A^{-1}b,A^{-1})$), in equilibrium, satisfies
\begin{equation}
\label{eq hmcexpect}
	\mathrm{E}[\alpha(x,y)] \rightarrow a(l) := 2 \Phi\left( -\frac{l^2 \sqrt{\tau}}{8} \right) 
\end{equation}
where $\tau = \lim_{d \rightarrow \infty} \frac{1}{d^{1+4\kappa}} \sum_{i=1}^d \lambda_i^4 \sin^2(\lambda_i T')$
and for eigenvector $q_i$ of $B$ corresponding to $\lambda_i^2$,
\begin{equation}
\label{eq hmcjump}
	\mathrm{E}[ |q_i^T V^{-1/2} (x' - x)|^2 ] = 2 \frac{1-\cos(\lambda_i T')}{\lambda_i^2} a(l) + \mathrm{o}\left( \frac{1-\cos(\lambda_i T')}{\lambda_i^2} \right)
\end{equation}
as $d \rightarrow \infty$, where $T' = Lh$.
\end{theorem}

The above result is an extension to the results in \cite{BPSSS2011,BPRSS2010} since their results only cover the situation when the target distribution has diagonal covariance, and $\kappa > 1/2$ or $\kappa = 0$.  

Note that in the above theory if we take $h = l d^{-r}$ for some $r \neq \tsfrac{1}{4}+\kappa$ then we would find that the expected acceptance rate would either tend to $0$ or $1$.  As is well known for Metropolis-Hastings algorithms, the optimal acceptance rate usually lies somewhere between $0$ and $1$ and taking $h = l d^{-1/4 - \kappa}$ is the correct scaling of $h$ to the dimension to achieve this.  However, even with this scaling we are still free to tune $h$ by varying $l$ and from the above result we can derive the optimal acceptance rate for HMC for a wider class of problems than previously studied in \cite{BPRSS2010,BPSSS2011}.  

For the efficiency of HMC, as well as considering the expected squared jump size of the chain, we must also consider the compute time per proposal which is proportional to $L = \lfloor \tsfrac{T}{h} \rfloor$ and depends on $d$.  Therefore, to maximise efficiency of HMC we must maximise the expected jump size divided by the compute time for a proposal.  With all other quantities held constant, this corresponds to varying $l$ to maximise
$$
	l a(l) = C \sqrt{s} \Phi(s) 
$$
where $s = \frac{l^2 \sqrt{\tau}}{8}$ and $C$ is a constant, which is maximised at $s_0 = 0.4250$, which corresponds to an expected acceptance rate of $2 \Phi(s_0) = 0.651$.  This is the same acceptance rate found in \cite{BPRSS2010} where they considered target distributions that are product distributions with $\lambda_i$ constant for all $i$.

Our theory goes further than earlier results in \cite{BPRSS2010} and \cite{BPSSS2011} for HMC because we also provide guidance on how to choose the other parameters in HMC: $V$ and $T$.  From Theorem \ref{thm hmc conv} we see that we should choose $V$ in a similar way to how we would choose a preconditioner for solving a linear system of equations.  We should choose $V$ to control the spread of eigenvalues of $VA$ so that $\kappa$ is as small as possible and $\tau$ is minimised, and we should choose $V$ so that the action of matrix multiplication with $V$ and sampling from $\normal(0,V^{-1})$ are cheap to compute.  This result is touched on in \cite{BPSSS2011} where they suggest taking $V = A^{-1}$, the perfect preconditioner, but if it is possible to sample from $\normal(0,A)$ then other sampling algorithms may be more efficient than HMC, particularly when the target distribution is $\normal(A^{-1}b,A^{-1})$.

Our theory also shows how to choose $T$ to maximise efficiency.  After tuning $h$ to achieve an acceptance rate of $0.651$ the expected squared jump size satisfies
$$
	\mathrm{E}[ |q_i^T V^{-1/2} (x' - x)|^2 ] \rightarrow 1.302 \frac{1-\cos(\lambda_i T')}{\lambda_i^2} \qquad \mbox{as $d \rightarrow \infty$}.
$$
By first deciding which $i$ correspond to directions we need to consider for our statistic of interest (this will depend on the eigenvectors $q_i$ of $AV$) we can then choose $T$ to maximise $1 - \cos(\lambda_i T')$ for those $i$.  Thus, how we choose $T$ depends on the problem, $V$ and our statistic of interest.  In the special case when $\lambda_i$ are equal then we should choose $T = \frac{\pi}{\lambda_i}$.


The theory presented here is for the leap-frog numerical integrator applied to the Hamiltonian system.  Higher order integrators are also suggested in \cite{BPRSS2010} and alternative numerical integrators based on splitting methods (in the ODEs context) are suggested in \cite{BCSS2014} that minimize the Hamiltonian error after $L$ steps of the integrator.  It may be possible to evaluate these other methods by first expressing them as a an AR(1) proposal and writing down the corresponding matrix splitting, then applying Theorems \ref{thm accept} and \ref{thm jumpsize} after a change of variables; as we have done for the leap-frog integrator.


For non-Gaussian target distributions, we cannot apply Theorems \ref{thm:ng1} and \ref{thm:ng2} to this HMC algorithm (with $L = \lfloor \tsfrac{T}{h} \rfloor$), since $\mathrm{E}_{\pi_d}[|y - x|_r^{2q}] \nrightarrow 0$ as $d \rightarrow \infty$, since $\tilde{g}_i^2 \lambda_i^{-1}, \; g_i \tilde{\lambda}_i^{-1} \nrightarrow 0$ as $d \rightarrow \infty$, see the proof of Theorem \ref{thm hmc conv} in the Appendix.  

\chapter{Concluding remarks}
\label{sec conclusion}

Until now, each MH algorithm with an AR(1) proprosal has required its own analysis, e.g. RWM, MALA, SLA, pCN and HMC.  In this article we have designed a unifying theory that encompasses all of these AR(1) proposals (except RWM) and other general AR(1) proposals where $G$ and $\Sigma$ are functions of $A$, for the case where the target distribution is a change of measure from a Gaussian reference measure.  


The main analysis tool we used is matrix splitting.  By writing an AR(1) proposal in terms of a matrix splitting, and requiring that the splitting matrices are functions of the target reference precision matrix $A$, then a simple change of variables diagonalises both the target reference $A$ and the matrices in the proposal; $G$ and $\Sigma$.  A consequence of this fact is that it is sufficient to analyse MH algorithms where the proposal and target reference are defined by diagonal matrices.  Essentially, we reduce the general case back to MH algorithms where the target reference measure is a product distribution, for which existing analysis of MALA, SLA and HMC can be extended to general AR(1) proposals.

In particular, we wrote down the obvious extension of results for Langevin proposals in \cite{BRS2009} to the case where the target reference measure is Gaussian with non-diagonal covariance and $\theta \in [0,1]$.   For HMC, we extended results in \cite{BPRSS2010} to the case when $\kappa \geq 0$ from $\kappa = 0$, we derived a new formula for the eigenvalues of the iteration matrix of the HMC proposal, and we provided criteria on how to choose $T$ and $V$ for HMC (previous analysis only said to adjust $h$ until the acceptance rate is $0.651$).

We also analysed a variation of the SLA algorithm where $L$ steps of the proposal are taken before the accept/reject step.  We simplified the evaluation of the acceptance probability using the surrogate transition method, and we found that, in the circumstances considered, it is optimal to take $L>1$ steps of SLA.  

The analysis presented here requires that the splitting matrices are functions of the target reference precision matrix.  In high dimensions this is a natural assumption to make because factorizing $A$ may be computationally infeasible.

Designing proposals for the MH algorithm to achieve efficient MCMC methods is a challenge, particularly for non-Gaussian target distributions, and the job is made harder by the difficultly we have in analysing the convergence properties of MH algorithms.  By focusing on AR(1) proposals in high dimension we have proven new theoretical results that provide us with criteria for evaluating and constructing new AR(1) proposals for efficient MH algorithms.  

Designing an efficient MH algorithm with an AR(1) proposal is often a balancing act between minimising the integrated autocorrelation time (we use maximising expected jump size as a proxy for this) and minimising compute time for each iteration of the chain.  If the proposal limit and target distributions are Gaussian and identical then it follows from the theory in \cite{F2013,FP2016,FP2014} that to construct an efficient AR(1) process we should try to satisfy the following conditions:
\begin{enumerate}
\item The spectral radius of $G$ should be as small as possible.
\item Computing an iteration of the stochastic AR(1) process should be as cheap.  This means that the action of $G$ and independent sampling from $\normal(0,\Sigma)$ should be cheap to compute.
\end{enumerate}
If the target distribution is Gaussian and different from the proposal distribution then Theorems \ref{thm accept} and \ref{thm jumpsize} suggest that, in addition, we should try to satisfy:
\begin{enumerate}
\setcounter{enumi}{2}
\item The difference between the target and proposal limit distributions should be as small as possible in the sense that the difference in means should be small, and the relative difference in precision matrix eigenvalues should be small.
\end{enumerate}
If the target distribution is non-Gaussian, then Theorems \ref{thm:ng1} and \ref{thm:ng2} suggest how we should try to satisfy:
\begin{enumerate}
\setcounter{enumi}{2}
\item The difference between the target reference and proposal limit distributions should be as small as possible in the sense that the difference in means should be small, and the relative difference in precision matrix eigenvalues should be small.
\end{enumerate}

In particular examples we can quantify these conditions using our theory.  For example, for proposals based on discretized generalised Langevin diffusion, Theorem \ref{thm genlang conv} shows us how the choice of symmetric positive definite matrix $V$ effects efficiency as it effects squared jump size in four ways.  Whilst choosing $V$ to maximise the limit in \eqref{eq jump2} (by minimising $\kappa$ and $\tau$) we should balance this against the scaling and direction that $V$ induces on $E[q_i^T V^{-1/2}(x'-x)^2]$ through $q_i$ and $V^{-1/2}$ on the left-hand side of \eqref{eq jump2}.    

Another example, pCN, satisfies conditions $1$ and $3$ above, but not necessarily condition $2$.  In particular, $G$ is the diagonal matrix with entries all $\tsfrac{1-h/4}{1+h/4}$ on the diagonal, and $\mathcal{A} = A$ and $\beta = b$.  However, each proposal for pCN requires an independent sample from $\normal(0,A^{-1})$, which may be infeasible in high dimensions.  In the special case when $A$ is diagonal, or a spectral decomposition of $A$ is available, then pCN satisfies all of our conditions for an efficient method.

Proposals for MALA and HMC are examples of proposals that are constructed by discretizing a stochastic differential equation that preserves the target distribution.  Our theory allows us to consider a wider selection of possible AR(1) proposals for the MH algorithm, that are not necessarily based on discretizing a stochastic differential equation.


\appendix
\chapter{Proofs}

\section{Proof of Lemma \ref{lem1}}

First note that
$$
	q(x,y) 
	\propto 
	\exp( \tsfrac{1}{2} (My - Nx - \beta)^T (M+N)^{-1} (My-Nx - \beta)).
$$
Simple algebra then yields
\begin{align*}
	2 \log \left( \frac{\pi_d(y)q(y,x)}{\pi_d(x)q(x,y)} \right) \!\! &= -y^T A y + x^T A x + 2b^T(y-x)\\
	&\qquad - (Mx-Ny-\beta)^T (M+N)^{-1} (Mx-Ny-\beta) \\
	&\qquad + (My-Nx-\beta)^T (M+N)^{-1} (My-Nx-\beta) \\
	&= -y^T A y + x^T A x + 2b^T(y-x)\\
	& \qquad - ((M-N)(x+y))^T (M+N)^{-1} ((M+N)(x - y)) \\
	& \qquad + 2 \beta^T(M+N)^{-1} ((M+N)(x-y)) \\
	&= -y^T (A-\mathcal{A}) y + x^T (A-\mathcal{A}) x + 2(b-\beta)^T(y-x).
\end{align*}	

\subsection{Proof of Lemma \ref{lem 5}}

Suppose $\lim_{d\rightarrow \infty} (\sum_{i=1}^d |t_i|^r)/ (\sum_{i=1}^d t_i^2 )^{r/2} = 0$.  Then for any $\epsilon > 0$ there exists a $D \in \mathbb{N}$ such that for any $d > D$,
$
	\sum_{i=1}^d |t_i|^r < \epsilon ( \sum_{i=1}^d t_i^2 )^{r/2}.
$
Then for any $k \in \mathbb{N}$, taking $\epsilon = 2^{-r/2}$, there exists a $D \geq k$ such that for any $d > D$,
$$
	|t_k|^r \leq \sum_{i=1}^d |t_i|^r < \frac{1}{2^{r/2}} \left( \sum_{i=1}^d t_i^2 \right)^{r/2}.
$$
Therefore, for any $d > D$, $t_k^2 < \tsfrac{1}{2} \sum_{i=1}^d t_i^2$ and so
$$
	\frac{\sum_{i=1, i \neq k}^d |t_i|^r}{\left( \sum_{i=1, i \neq k}^d t_i^2 \right)^{r/2}}  
	< \frac{ \sum_{i=1}^d |t_i|^r }{\left(\tsfrac{1}{2} \sum_{i=1}^d t_i^2 \right)^{r/2}} 
	= 2^{r/2} \frac{ \sum_{i=1}^d |t_i|^r }{\left(\sum_{i=1}^d t_i^2 \right)^{r/2}}.
$$

\section{Proof of Theorem \ref{thm genlang conv}}

We use the following technical lemma in the proof of Theorem \ref{thm genlang conv}.
\begin{lemma} \label{lem 6a} 
Suppose $\{t_i\} \subset \bbR$ is a sequence such that $0 < t_i \leq C d^{-1/3} (\tsfrac{i}{d})^{2\kappa}$ for $C>0$ and $\kappa \geq 0$.  If $s > 3$, then $
\lim_{d\rightarrow \infty} \sum_{i=1}^d t_i^s = 0
$.
\end{lemma}

\begin{proof}
\begin{align*}
	\lim_{d \rightarrow \infty} \sum_{i=1}^d t_i^s 
	\leq C^s \lim_{d \rightarrow \infty}  d^{1-s/3} \sum_{i=1}^d \tsfrac{1}{d} \left( \tsfrac{i}{d} \right)^{2\kappa s} 
	= C^s \lim_{d \rightarrow \infty}  d^{1-s/3} \int_0^1 z^{2\kappa s} \mathrm{d}z 
	= 0.
\end{align*}
\end{proof}

\emph{Proof of Theorem \ref{thm genlang conv}}.  First note that $VA$ and $V^{1/2}AV^{1/2}$ are similar, so they have the same eigenvalues.  Lemma \ref{lem ch1} implies that it is equivalent to study the MH algorithm with target and proposal given by \eqref{eq hat1}.  We now attempt to apply Theorems \ref{thm:ng1} and \ref{thm:ng2} to \eqref{eq hat1}.   Note that the splitting matrices for \eqref{eq hat1} are functions of $B = V^{1/2}A V^{1/2}$.  Let $t_{d,i}=h \lambda_i = \mathcal{O}(d^{-t})$ for $t = 1/3+\kappa$, let $s_i = h \lambda_i^2 = l d^{-1/3} (\tsfrac{\lambda_i}{d})^{2\kappa} = \mathcal{O}(d^{-1/3})$ and let $\rho = (\theta-\tsfrac{1}{2})/2$.  Then
\begin{align*}
	\tilde{\lambda}_i^2 &= (1+\rho s_i) \lambda_i^2, & G_i &= 1 - \frac{\frac{1}{2}s_i}{1 + \frac{\theta}{2}s_i}, &
	\tilde{g}_i &= \frac{\frac{1}{2}s_i}{1 + \frac{\theta}{2}s_i}, & g_i &= \frac{s_i (1 + \rho s_i)}{1 + \frac{\theta}{2}s_i}, \\
	r_i &= -\rho s_i, & \tilde{r}_i &= \frac{1}{1 + \rho s_i}, & \hat{r}_i &= 0,
\end{align*}
so that
\begin{align*}
	T_{0i} &= T_{1i} = T_{2i} = 0, \\
	T_{3i} &= \frac{-\frac{1}{2} \rho s_i^2 (1 + \rho s_i)}{(1 + \frac{\theta}{2}s_i)^2}, &
	T_{4i} &= \frac{\frac{1}{2} \rho s_i^2 }{(1 + \frac{\theta}{2}s_i)^2}, &
	T_{5i} &= \frac{\frac{1}{2} \rho s_i^{3/2} (1-\frac{1-\theta}{2} s_i) }{(1 + \frac{\theta}{2}s_i)^2}.
\end{align*}
Hence 
$$
	\tilde{g}_i^2 \hat{r}_i^2 \lambda_i = 0, \qquad \tilde{g}_i^2 \lambda_i^{-1} = \mathcal{O}(t_{d,i} s_i) = \mathrm{o}(t_{d,i}), \qquad g_i \tilde{\lambda}_i^{-1} = \mathcal{O}(t_{d,i}), \qquad \tilde{r}_i = \mathcal{O}(1)
$$ 
as $d \rightarrow \infty$, and $T_{3i} = \mathcal{O}(d^{-2/3})$, $T_{4i} = \mathcal{O}(d^{-2/3})$ and $T_{5i} = \mathcal{O}(d^{-1/2})$ as $d \rightarrow \infty$.


We also have
\begin{align*}
	\mu 
	&= \lim_{d \rightarrow \infty} \sum_{i=1}^d \mu_i 
	= \lim_{d \rightarrow \infty} \sum_{i=1}^d T_{3i} + T_{4i} \\
	&= \lim_{d \rightarrow \infty} -\tsfrac{\rho^2}{2} \sum_{i=1}^d \frac{s_i^3}{(1 + \frac{\theta}{2}s_i)^2} 
	= \lim_{d \rightarrow \infty} -\tsfrac{\rho^2}{2} \sum_{i=1}^d s_i^3 
	= - \frac{l^6 (\theta-\frac{1}{2})^2 \tau}{8},
\end{align*}
and using Lemma \ref{lem 6a},
\begin{align*}
	\sigma^2
	&= \lim_{d \rightarrow \infty} \sum_{i=1}^d \sigma_i^2 
	= \lim_{d \rightarrow \infty} \sum_{i=1}^d 2 T_{3i}^2 + 2 T_{4i}^2 + T_{5i}^2 \\
	&= \lim_{d \rightarrow \infty} \sum_{i=1}^d \frac{1}{(1+\frac{\theta}{2}s_i)^4} \left( \tsfrac{1}{2}\rho^2 s_i^4 (1+\rho s_i)^2 + \tsfrac{1}{2} \rho^2 s_i^4 + \tsfrac{1}{4} \rho^2 s_i^3 (1-\tsfrac{1-\theta}{2}s_i)^2 \right)  \\
	&= \lim_{d \rightarrow \infty} \sum_{i=1}^d \left( \tsfrac{1}{2}\rho^2 s_i^4 (1+\rho s_i)^2 + \tsfrac{1}{2} \rho^2 s_i^4 + \tsfrac{1}{4} \rho^2 s_i^3 (1-\tsfrac{1-\theta}{2}s_i)^2 \right) \\
	&= \lim_{d \rightarrow \infty} \sum_{i=1}^d \tsfrac{1}{4} \rho^2 s_i^3 \\
	&= \lim_{d \rightarrow \infty} \frac{l^6 (\theta - \tsfrac{1}{2})^2 \tau}{4}.
\end{align*}
Hence $\lim_{d\rightarrow\infty} \sum_{i=1}^d T_{1i}^2 + T_{2i}^2 + 2 T_{3i}^2 + 2 T_{4i}^2 + T_{5i}^2 < \infty$.  Also note that, by Lemma \ref{lem 6a},
$$
	\lim_{d\rightarrow \infty} \sum_{i=1}^d T_{1i}^2 + T_{3i}^2 = \lim_{d\rightarrow \infty} \sum_{i=1}^d T_{3i}^2 = \lim_{d\rightarrow \infty} \sum_{i=1}^d \tsfrac{1}{4} \rho^2 s_i^4 = 0.
$$	
It follows that $\frac{\mu}{\sigma} = -\sigma - \frac{\mu}{\sigma} = -l^3 |\theta-\frac{1}{2}| \sqrt{\tau}/4$ and $\mu + \frac{\sigma^2}{2} = 0$.  Hence we obtain \eqref{eq expect} from Theorem \ref{thm:ng1}, using Corollary \ref{cor:ng}.

For the expected jump size, first note that $\hat{r}_i = 0$.  Also,
$$
	\frac{g_i}{\tilde{\lambda}_i^2}  = \frac{h}{1 + \tsfrac{\theta}{2} s_i} = h + \mathcal{O}(h s_i) =h + \mathrm{o}(h), \qquad \mbox{as $d \rightarrow \infty$},
$$
and using the fact that $\gamma_{d,i}^2$ is uniformly bounded (see proof of Theorem \ref{thm:ng1}),
$$
	\frac{\tilde{g}_i^2 \gamma_{d,i}^{1/2}}{\lambda_i^2} = \mathcal{O}\left(\frac{s_i^2}{\lambda_i^2} \right) = \mathcal{O}(h s_i) = \mathrm{o}(h) \qquad \mbox{as $d \rightarrow \infty$}.
$$
Therefore, applying Theorem \ref{thm:ng2} to the MH algorithm \eqref{eq hat1} we find
$$
	\mathrm{E}_{\pi_d}[(q_i^T(x_i' - x_i))^2] = 2 h \Phi\left( -\frac{l^3 |\theta-\tsfrac{1}{2}| \sqrt{\tau}}{4} \right) + \mathrm{o}(h)
$$
as $d \rightarrow \infty$, where $\pi_d$ is given in \eqref{eq hat1}.  Reversing the coordinate transformation $x \leftrightarrow V^{-1/2} x$ we obtain \eqref{eq jump2}.

\section{Proof of Lemma \ref{lem:lstepalpha}}

By replacing $\pi_d$ with $\pi^*$ in the proof of Lemma \ref{lem1} it follows that
$$
	\pi^*(x) q(x,y) = \pi^*(y) q(y,x)
$$
where $q$ corresponds the transition kernel for the SLA proposal.  That is, $q(\cdot,\cdot)$ satisfies detailed balance with respect to $\pi^*(\cdot)$.  It then follows from 
$$
	q_L(x,y) = \idotsint q(x,y^{(1)}) q(y^{(1)},y^{(2)}) \dotsm q(y^{(L-1)},y) \; \dd y^{(1)} \dd y^{(2)} \dotsm \dd y^{(L-1)}
$$
that $q_L(\cdot,\cdot)$ satisfies detailed balance with respect to $\pi^*(\cdot)$, so that
$$
	\pi^*(x) q_L(x,y) = \pi^*(y) q_L(y,x)
$$
Using this fact, we find
$$
	\frac{\pi_d(y)q_L(y,x)}{\pi_d(x) q_L(x,y)} = \frac{ \frac{\pi_d(y)}{\pi^*(y)} \pi^*(y) q_L(y,x) }{ \frac{\pi_d(x)}{\pi^*(x)} \pi^*(x) q_L(x,y)}
	= \frac{ \frac{\pi_d(y)}{\pi^*(y)}\pi^*(y) q_L(y,x) }{ \frac{\pi_d(x)}{\pi^*(x)} \pi^*(y) q_L(y,x) }
	= \frac{ \pi_d(y) \pi^*(x) }{ \pi_d(x) \pi^*(y)}.
$$
Hence, result.

\section{Proof of Theorem \ref{thm:lstep}}

We prove this theorem by applying Theorems \ref{thm:ng1} and \ref{thm:ng2} to $L$-step SLA.  First, let us check the conditions for these theorems.  Let $G = I-\tsfrac{h}{2}A$, $\Sigma = hI$ be the proposal matrices for the SLA proposal.  These matrices are functions of $A$.  It follows that $G_L$ and $\Sigma_L$ for the $L$-step SLA proposal are also functions of $A$ and $G_L \Sigma_L$ is symmetric.  It then follows from Corollary \ref{lem equiv} that the splitting matrices for $L$-step SLA are also functions of $A$.

Let $t_{d,i} = h \lambda_i = \mathcal{O}(d^{-1/3-\kappa})$ and $s_i = h \lambda_i^2 = \mathcal{O}(d^{-1/3})$ as $d \rightarrow \infty$.  Since $\mathcal{A}$ and $\beta$ for $L$-step SLA are the same as for SLA, from the proof of Theorem \ref{thm genlang conv} we have $\tilde{\lambda}_i^2 = (1-\frac{1}{4}s_i)\lambda_i^2$, $\hat{r}_i = 0$, $r_i = \tsfrac{1}{4} s_i$ and $\tilde{r}_i = (1-\tsfrac{1}{4}s_i)^{-1} = \mathcal{O}(1)$.

Also, $G_i = 1 - \tsfrac{1}{2} s_i$ so the eigenvalues of $G_L$ are $G_i^L = (1- \tsfrac{1}{2}s_i)^L = 1 - \tsfrac{L}{2} s_i + \mathcal{O}(s_i^2)$.  Hence, $\tilde{g}_{Li} = 1 - G_i^L = \tsfrac{L}{2} s_i + \mathcal{O}(s_i^2)$ and $g_{Li} = 1 - G_i^{2L} = L s_i + \mathcal{O}(s_i^2)$.

It then follows that $\tilde{g}_i^2 \hat{r}_i^2 \lambda_i = 0$, $\tilde{g}_{Li}^2 \lambda_i^{-1} = \mathcal{O}(s_i^2 \lambda_i^{-1}) = \mathcal{O}(s_i t_{d,i}) = \mathrm{o}(t_{d,i})$, and $g_{Li} \tilde{\lambda}_i^{-1} = \mathcal{O}(t_{d,i})$.

We also have for $L$-step SLA,
\begin{align*}
	T_{3i} = \tsfrac{1}{8} L s_i^2 + \mathcal{O}(s_i^3), \quad 
	T_{4i} = - \tsfrac{1}{8} L s_i^2 + \mathcal{O}(s_i^3), \quad
	T_{5i} = -\tsfrac{1}{4} s_i (L s_i + \mathcal{O}(s_i^2))^{1/2} (1 + \mathcal{O}(s_i)),
\end{align*}
so $T_{3i} = \mathcal{O}(d^{-2/3})$, $T_{4i} = \mathcal{O}(d^{-2/3})$ and $T_{5i} = \mathcal{O}(d^{-1/2})$.

We can also calculate, using Lemma \ref{lem 6a}, 
$$
	\lim_{d\rightarrow \infty} \sum_{i=1}^d T_{3i}^2 = \tsfrac{1}{64} L^2 \lim_{d \rightarrow \infty} \sum_{i=1}^d s_i^4 = 0,
$$
and similarly, $\lim_{d \rightarrow \infty} \sum_{i=1}^d T_{4i}^2 = 0$, and as in the proof of Theorem \ref{thm genlang conv} (calculating $\sigma^2$),
$$
	\lim_{d \rightarrow \infty} \sum_{i=1}^d T_{5i}^2 = \lim_{d \rightarrow \infty} \sum_{i=1}^d \tsfrac{1}{16} L s_i^3 = L \frac{l^6 \tau}{16} =: \sigma_L^2,
$$
Using Lemma \ref{lem 6a}, and as for calculating $\mu$ in the proof of Theorem \ref{thm genlang conv},
\begin{align*}
	\mu_L &= \lim_{d \rightarrow \infty} \sum_{i=1}^d T_{3i} + T_{4i} = \lim_{d \rightarrow \infty} \sum_{i=1}^d -\tsfrac{1}{2} r_i g_i (\tilde{r}_i-1) \\
	&= \lim_{d \rightarrow \infty} \sum_{i=1}^d -\tsfrac{1}{32}L s_i^3 + \mathcal{O}(s_i^4) \\
	&= \lim_{d \rightarrow \infty} \sum_{i=1}^d -\tsfrac{1}{32}L s_i^3 \\
	&= - L \frac{l^6 \tau}{32}.
\end{align*}
Hence, we obtain \eqref{eq:lstepa} from Theorem \ref{thm:ng1} and Corollary \ref{cor:ng}.

For the expected jump size, \eqref{eq ljump} follows from $\hat{r}_i = 0$, $g_{Li} \tilde{\lambda}_i^{-2} = L h + \mathrm{o}(h)$, $\tilde{g}_{Li}^2 \lambda_i^{-2} = \mathrm{o}(h)$, and $\gamma_{d,i}$ bounded uniformly (see proof of Theorem \ref{thm:ng1}).

\section{Proof of Theorem \ref{thm HMC1}}

The result will follow from Corollary \ref{lem equiv} with $G = (K^L)_{11}$ and $\Sigma = (K^L)_{12} V^{-1} (K^L)_{12}^T$ but we must first check that $G \Sigma$ is symmetric.  Define $\mathcal{V}$, $\mathcal{K}$ and $B$ as in Theorem \ref{thm iso hmc}.  Then $K = \mathcal{V} \mathcal{K} \mathcal{V}^{-1}$, so that $K^L = \mathcal{V} \mathcal{K}^L \mathcal{V}^{-1}$ and
$$
	(K^L)_{11} = V^{1/2} (\mathcal{K}^L)_{11} V^{-1/2} 
	\quad \mbox{and} \quad
	(K^L)_{12} = V^{1/2} (\mathcal{K}^L)_{12} V^{1/2}.
$$
Then
\begin{align*}
	G \Sigma = V^{1/2} (\mathcal{K}^L)_{11} (\mathcal{K}^L)_{12}^{2} V^{1/2} 
\end{align*}
which is symmetric because $V$ and $B$ are symmetric and $(\mathcal{K}^L)_{11}$ and $(\mathcal{K}^L)_{12}$ are polynomials of $B$.

\section{Proof of Corollary \ref{cor 15}}

First note that 
$$
	(I-(K^L)_{11}) \mathcal{A}^{-1} \beta = \left( SJ \left[ \twobyone{0}{\tsfrac{h}{2}b} \right] \right)_1,
$$
so we are required to show that
$$
	(I-(K^L)_{11}) A^{-1} b = \left( SJ \left[ \twobyone{0}{\tsfrac{h}{2}b} \right] \right)_1,
$$
which holds if
$$
	(I-K^L) \left[ \twobyone{A^{-1}b}{0} \right] = SJ \left[ \twobyone{0}{\tsfrac{h}{2}b} \right].
$$
Using $S = (I-K)^{-1}(I-K^L)$, we can equivalently show
$$
	(I-K) \left[ \twobyone{A^{-1}b}{0} \right] = J \left[ \twobyone{0}{\tsfrac{h}{2}b} \right],
$$
which is easy to check.

\section{Proof of Theorem \ref{thm uhmc conv}}

Define a spectral decomposition 
\begin{equation}
\label{eq specdecomp}
	V^{1/2}AV^{1/2} = Q \Lambda Q^T
\end{equation}
where $Q$ is an orthogonal matrix and $\Lambda = \operatorname{diag}(\lambda_1^2,\dotsc,\lambda_d^2)$ is a diagonal matrix of eigenvalues of $V^{1/2}AV^{1/2}$ ($VA$ is similar to $V^{1/2}AV^{1/2}$ so they have the same eigenvalues).  Also define $\mathcal{V}$ as in Theorem \ref{thm iso hmc} and 
$$
	\tilde{Q} = 	\left[ \twobytwo{Q}{0}{0}{Q} \right] \in \bbR^{2d \times 2d}.
$$
A similarity transform of $K$ is defined by
$$
	K = \mathcal{V} \tilde{Q} \tilde{K} \tilde{Q}^T \mathcal{V}^{-1}
\quad \mbox{
with} \quad
	\tilde{K} = \left[ \twobytwo{I - \tsfrac{h^2}{2}\Lambda}{hI}{-h\Lambda + \tsfrac{h^3}{4} \Lambda^2}{I - \tsfrac{h^2}{2}\Lambda} \right].
$$
Hence $K$ and $\tilde{K}$ have the same eigenvalues.  Moreover, $K^L = \mathcal{V} \tilde{Q} \tilde{K}^L \tilde{Q}^T \mathcal{V}^{-1}$ and it follows that 
$$
	(K^L)_{11} = V^{1/2} Q (\tilde{K}^L)_{11} Q^T V^{-1/2}.
$$
Thus $(K^L)_{11}$ and $(\tilde{K}^L)_{11}$ are similar.

Notice that $\tilde{K}$ is a $2\times2$ block matrix where each $d\times d$ block is diagonal.  Therefore, $\tilde{K}^L$ is also a $2\times2$ block matrix with diagonal blocks.  In particular, $(\tilde{K}^L)_{11}$ is a diagonal matrix, so the eigenvalues of $(\tilde{K}^L)_{11}$ are on the diagonal of $(\tilde{K}^L)_{11}$.  Moreover, 
$$
	[(\tilde{K}^L)_{11}]_{ii} = (k_i^L)_{11}
$$
where $[(\tilde{K}^L)_{11}]_{ii}$ is the $i^{\mathrm{th}}$ diagonal entry of $(\tilde{K}^L)_{11}$,  $(k_i^L)_{11}$ is the $(1,1)$ entry of the matrix $k_i^L \in \bbR^{2\times2}$, and $k_i \in \bbR^{2\times2}$ is defined by
$$
	k_i = \left[ \twobytwo{(\tilde{K}_{11})_{ii}}{(\tilde{K}_{12})_{ii}}{(\tilde{K}_{21})_{ii}}{(\tilde{K}_{22})_{ii}} \right]
	= \left[ \twobytwo{1-\tsfrac{h^2}{2}\lambda_i^2}{h}{-h\lambda_i^2 + \tsfrac{h^3}{4}\lambda_i^4}{1-\tsfrac{h^2}{2}\lambda_i^2} \right].
$$  
The matrix $k_i$ can be factorized
$$
	k_i = \left[ \twobytwo{1}{0}{0}{a} \right] \left[ \twobytwo{\cos(\theta_i)}{-\sin(\theta_i)}{\sin(\theta_i)}{\cos(\theta_i)} \right]
	\left[ \twobytwo{1}{0}{0}{a^{-1}} \right]
$$
where $a = \lambda \sqrt{1 - \tsfrac{h^2}{4}\lambda_i^2}$ and $\theta_i = -\cos( 1-\tsfrac{h^2}{2}\lambda_i^2)$.  Therefore,
$$
	k_i^L = \left[ \twobytwo{1}{0}{0}{a} \right] \left[ \twobytwo{\cos(L \theta_i)}{-\sin(L \theta_i)}{\sin(L \theta_i)}{\cos(L \theta_i)} \right]
	\left[ \twobytwo{1}{0}{0}{a^{-1}} \right]
$$
and hence
$$
	[(\tilde{K}^L)_{11}]_{ii} = (k_i^L)_{11} = \cos(L\theta_i).
$$

\section{Proof of Theorem \ref{thm hmc conv}}

Theorem \ref{thm iso hmc} implies that it is equivalent to study the MH algorithm with target and proposal given by \eqref{eq hmchat1}, and since $(\mathcal{K}^L)_{11}$ and $(\mathcal{K}^L)_{12}$ are functions of $B$ we can apply Theorems \ref{thm accept} and \ref{thm jumpsize}.  Using the spectral decomposition \eqref{eq specdecomp} note that
$$
	(\mathcal{K}^L)_{ij} = Q (\tilde{K}^L)_{ij} Q^T \qquad \mbox{for $i,j=1,2$}.
$$
where $\tilde{K}$ is defined in the proof of Theorem \ref{thm uhmc conv}, and where it is shown that $(\tilde{K}^L)_{ij}$ is diagonal and
$$
	[(\tilde{K}^L)_{11}]_{ii} = \cos (L \theta_i)
$$
where $\theta_i = -\cos^{-1}(1-\tsfrac{h^2}{2}\lambda_i^2)$.  Similarly,
$$
	[(\tilde{K}^L)_{12}]_{ii} = -a_i^{-1} \sin (L \theta_i)
$$
where $a_i = \lambda_i \sqrt{1-\tsfrac{h^2}{4}\lambda_i^2}$.  Moreover, $\mathcal{A} = Q (\tilde{K}^L)_{12}^2 (I-(\tilde{K}^L)_{11}^2) Q^T$ so that $\tilde{\lambda}^2_i = -a_i^2$ and if we let $s_i = h^2 \lambda_i^2 = l^2 d^{-1/2} (\frac{\lambda_i}{d^\kappa})^2 = \mathcal{O}(d^{-1/2})$, then 
\begin{align*}
	\tilde{\lambda}_i^2 &= \lambda_i^2 (1 - \tsfrac{1}{4}s_i), & G_i &= \cos(L\theta_i), & \tilde{g}_i &= 1 - \cos(L\theta_i), & g_i &= \sin^2(L \theta_i), \\
	r_i &= \tsfrac{1}{4} s_i, & \tilde{r}_i &= \tsfrac{1}{1-\tsfrac{1}{4}s_i}, & \hat{r}_i &= 0.
\end{align*}
Note that we used Corollary \ref{cor 15} to show $\hat{r}_i = 0$.  Then
\begin{align*}
	T_{0i} &= T_{1i} = T_{2i} = 0, \\
	T_{3i} &= \tsfrac{1}{8} s_i \sin^2(L\theta_i), &
	T_{4i} &= - \frac{\tsfrac{1}{8} s_i \sin^2(L\theta_i)}{1-\tsfrac{1}{4}s_i}, &
	T_{5i} &= - \frac{ \tsfrac{1}{8} s_i \sin(2L\theta_i)}{\sqrt{1-\tsfrac{1}{4}s_i}}.
\end{align*}
Using the trigonmetric expansion $\cos^{-1}(1-z) = \sqrt{2z} + \mathcal{O}(z^{3/2})$, and defining $T'$ such that $L = \tsfrac{T'}{h}$ we find 
$$
	L \theta_i
			= L ( - \sqrt{s_i} + \mathcal{O}(s_i^{3/2}) )
			= - L s_i^{1/2} (  1 + \mathcal{O}(s_i) )
			= - T' \lambda_i ( 1 + \mathcal{O}(d^{-1/2}) )
$$
hence, there exists a function $T''(d)$ such that $L \theta_i = - T'' \lambda_i$ and
$T''(d) = T' + \mathcal{O}(d^{-1/2})$.

To apply Theorems \ref{thm accept} and \ref{thm jumpsize} we need to check \eqref{eq Tcond}.  For some $h > 0$, $c l^2 (\tsfrac{i}{d})^{2\kappa} \leq d^{1/2} s_i = l^2 (\tsfrac{\lambda_i}{d^\kappa})^2 \leq C l^2 (\tsfrac{i}{d})^{2\kappa}$ and so we find
\begin{align*}
	\lim_{d \rightarrow \infty} \frac{ \sum_{i=1}^d |T_{3i}|^{2+\delta} }{ \left( \sum_{i=1}^d |T_{3i}|^2 \right)^{1+\delta/2} }
	&= \lim_{d \rightarrow \infty} \frac{ \sum_{i=1}^d |s_i \sin^2(L\theta_i) |^{2+\delta} }{ \left( \sum_{i=1}^d |s_i \sin^2(L\theta_i)|^2 \right)^{1+\delta/2} } \\
	&= \lim_{d \rightarrow \infty} d^{-\delta/2} \left( \frac{ \left( \sum_{i=1}^d | d^{1/2} s_i \sin^2(T'' \lambda_i) |^{2+\delta} \right)^{1/(2+\delta)} }{ \left( \sum_{i=1}^d |d^{1/2} s_i \sin^2(T'' \lambda_i)|^2 \right)^{1/2} } \right)^{2+\delta} \\
	&= 0
\end{align*}
since $\nrm{v}_q \leq \nrm{v}_p$ for $q \geq p> 0 $ for $l^p$-norms on $\bbR^d$.
Similar arguments verify \eqref{eq Tcond} for $T_{4i}$ and $T_{5i}$.  Now we can apply Theorem \ref{thm accept} with
\begin{align*}
	\mu 
	&= \lim_{d \rightarrow \infty} \sum_{i=1}^d T_{3i} + T_{4i} \\
	&= -\frac{1}{32} \lim_{d \rightarrow \infty}  \sum_{i=1}^d \frac{s_i^2 \sin^2(L\theta_i)}{1-\tsfrac{1}{4}s_i} \\
	&= -\frac{1}{32} \lim_{d \rightarrow \infty}  \sum_{i=1}^d s_i^2 \sin^2(L\theta_i) \\
	&= -\frac{l^4}{32} \lim_{d \rightarrow \infty} \frac{1}{d^{1+4\kappa}} \sum_{i=1}^d \lambda_i^4 \sin^2(\lambda_i T''(d)) \\
	&= -\frac{l^4}{32} \lim_{d \rightarrow \infty} \frac{1}{d^{1+4\kappa}} \sum_{i=1}^d \lambda_i^4 \sin^2(\lambda_i T') \\
	&= -\frac{l^4 \tau}{32}
\end{align*}
and similarly,
\begin{align*}
	\sigma^2 &= \lim_{d \rightarrow \infty} \sum_{i=1}^d 2 T_{3i}^2 + 2 T_{4i}^2 + T_{5i}^2 \\
	&= \lim_{d \rightarrow \infty} \sum_{i=1}^d \frac{1}{32} s_i^2 \sin^4(L\theta_i) + \frac{1}{32} \frac{s_i^2 \sin^4(L\theta_i)}{(1-\tsfrac{1}{4}s_i)^2} + \frac{1}{64} \frac{s_i^2 \sin^2(2L\theta_i)}{1-\tsfrac{1}{4}s_i} \\
	&= \lim_{d \rightarrow \infty} \sum_{i=1}^d \frac{1}{16} s_i^2 \sin^4(L\theta_i) + \frac{1}{64} s_i^2 \sin^2(2L\theta_i) \\
	&= \frac{1}{16} \lim_{d \rightarrow \infty} \sum_{i=1}^d s_i^2 \sin^2(L\theta_i) \\
	&= \frac{l^4 \tau}{16}.
\end{align*}
Hence $\frac{\mu}{\sigma} = -\sigma - \frac{\mu}{\sigma} = - \frac{l^2 \sqrt{\tau}}{8}$ and $\mu+\sigma^2/2 = 0$, so from Theorem \ref{thm accept} we obtain \eqref{eq hmcexpect}.

For the expected jump size, we apply Theorem \ref{thm jumpsize} with
\begin{align*}
	U_1 &= \frac{\tilde{g}_i^2}{\lambda_i^2} + \frac{g_i}{\tilde{\lambda}_i^2} = \frac{(1-\cos(L\theta_i))^2}{\lambda_i^2} + \frac{\sin^2(L\theta_i)}{\lambda_i^2 (1-\frac{1}{4}s_i)} 
	= \frac{2(1-\cos(\lambda_i T'))}{\lambda_i^2} + \mathcal{O}\left( \frac{s_i}{\lambda_i^2} \right) \\ 
	&= \frac{2(1-\cos(\lambda_i T'))}{\lambda_i^2} + \mathcal{O}\left( h \right)
\end{align*}
as $d\rightarrow \infty$.  Also, it is straightforward to show that $\mu_i = \mathcal{O}(d^{-1})$ and $\sigma_i^2 = \mathcal{O}(d^{-1})$, so $\mu^-$ and $\sigma^-$ exist.  Recall that $Z \sim \normal(\mu,\sigma^2)$ and $X \sim \normal(\mu^-,(\sigma^-)^2)$, then $X-Z \sim \normal(-\mu_i,\sigma^2 + (\sigma^-)^2)$, and using the fact that $z \mapsto 1 \wedge \mathrm{e}^z$ is monotone and globally Lipschitz continuous
$$
	|U_2 - \mathrm{E}[\alpha(x,y)]| = |\mathrm{E}[ 1 \wedge \mathrm{e}^X] - \mathrm{E}[1 \wedge \mathrm{e}^Z]| \leq |\mathrm{E}[ X-Z]| = |\mu_i| = \mathcal{O}(d^{-1}).
$$
Then $U_2 = \mathrm{E}[\alpha(x,y)] + \mathcal{O}(d^{-1})$ as $d \rightarrow \infty$, and 
\begin{align*}
	|U_3| &= (\sigma_i^2 + \mu_i^2)^{1/2} \frac{\sqrt{3}}{\lambda_i^2} \left( \tilde{g}_i^2 + \tilde{r}_i g_i \right) \\
	&= (\sigma_i^2 + \mu_i^2)^{1/2} \frac{\sqrt{3}}{\lambda_i^2} \left( (1-\cos(L\theta_i))^2 + \frac{1}{1-\frac{1}{4}s_i} \sin^2(L\theta_i) \right) \\
	&= \mathcal{O}(d^{-1/2}).
\end{align*}
Therefore, we obtain \eqref{eq hmcjump}.

\bibliographystyle{plain}
\bibliography{paper8bib} 

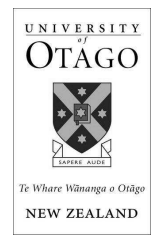
\end{document}